%% file: main.tex
\documentclass{article}

\PassOptionsToPackage{table}{xcolor}
\usepackage[utf8]{inputenc}
\usepackage[T1]{fontenc}
\usepackage{report}
\usepackage{float}

\usepackage{microtype}
\usepackage{graphicx}
\usepackage{wrapfig}
\usepackage{subcaption}
\usepackage{booktabs}
\usepackage{multirow}
\usepackage{amsmath}
\usepackage{amssymb}
\usepackage{mathtools}
\usepackage{amsthm}
\usepackage{algorithm}
\usepackage{algorithmic}
\usepackage{xcolor}
\usepackage{colortbl}
\usepackage{placeins}
\usepackage{titletoc}
\usepackage{hyperref}
\hypersetup{
  pdftitle={SOLAR: SVD-Optimized Lifelong Attention for Recommendation},
  pdfauthor={Chenghao Zhang, Chao Feng, Yuanhao Pu, Xunyong Yang, Wenhui Yu, Xiang Li, Chunjie Chen, Kaiqiao Zhan},
  pdfsubject={Lifelong recommendation with low-rank softmax attention},
  pdfkeywords={SOLAR, SVD-Attention, lifelong recommendation, efficient attention}
}
\usepackage{url}
\usepackage[capitalize,noabbrev]{cleveref}

\usepackage{tcolorbox}
\tcbuselibrary{breakable}
\definecolor{solarOrange}{RGB}{230,161,91}
\definecolor{solarOrangeLight}{RGB}{255,232,204}
\newtcolorbox{abstractbox}{
  colback=solarOrangeLight,
  colframe=solarOrange,
  boxrule=1pt,
  arc=4mm,
  left=8pt,
  right=8pt,
  top=8pt,
  bottom=8pt,
  opacityback=0.95,
  breakable
}
\date{}

\theoremstyle{plain}
\newtheorem{theorem}{Theorem}[section]
\newtheorem{proposition}[theorem]{Proposition}
\newtheorem{lemma}[theorem]{Lemma}
\newtheorem{corollary}[theorem]{Corollary}
\theoremstyle{definition}
\newtheorem{definition}[theorem]{Definition}

\theoremstyle{remark}

\title{SOLAR: SVD-Optimized Lifelong Attention\\ for Recommendation}

\DeclareRobustCommand{\solarauthors}{%
\begingroup
\hypersetup{urlcolor=black}
\setlength{\tabcolsep}{9pt}
\renewcommand{\arraystretch}{1.3}
\begin{tabular}{cccc}
\href{https://openreview.net/profile?id=~Chenghao_Zhang8}{Chenghao Zhang}$^{1,*}$ &
\href{https://openreview.net/profile?id=~Chao_Feng3}{Chao Feng}$^{1,*}$ &
\href{https://openreview.net/profile?id=~Yuanhao_Pu1}{Yuanhao Pu}$^{1,*,\dagger}$ &
\href{https://openreview.net/profile?id=~Xunyong_Yang1}{Xunyong Yang}$^{1}$ \\
\href{https://openreview.net/profile?id=~Wenhui_Yu3}{Wenhui Yu}$^{2}$ &
\href{https://openreview.net/profile?id=~Xiang_Li49}{Xiang Li}$^{1}$ &
\href{https://openreview.net/profile?id=~Chunjie_Chen1}{Chunjie Chen}$^{1}$ &
\href{https://openreview.net/profile?id=~Kaiqiao_Zhan1}{Kaiqiao Zhan}$^{1}$
\end{tabular}
\endgroup
}
\author{\solarauthors}

\begin{document}

\maketitle

% Affiliation and contribution notes follow the supplied template.
\begin{NoHyper}
\begingroup
\renewcommand{\thefootnote}{}
\footnotetext{$^{1}$~Kuaishou Technology.}
\footnotetext{$^{2}$~Unaffiliated.}
\footnotetext{$^{*}$~Equal contribution.}
\footnotetext{$^{\dagger}$~Work done during an internship at Kuaishou Technology.}
\endgroup
\end{NoHyper}

\begin{abstract}
\begin{abstractbox}
Attention mechanism remains the defining operator in Transformers since it provides expressive global credit assignment, yet its quadratic cost in sequence length $N$ makes long-context modeling expensive and often forces truncation or other heuristics. Linear attention reduces complexity to $\mathcal{O}(N d^2)$ by reordering computation through kernel feature maps, but this reformulation drops the softmax mechanism and shifts the attention score distribution. Lifelong recommendation requires efficient attention as well, for large-scale sequence modeling with user histories and candidate items under tight latency and resource constraints. We introduce \textbf{SVD-Attention}, a novel attention mechanism, and \textbf{SOLAR}, a set-aware framework built on it for lifelong recommendation. SVD-Attention factorizes low-rank embeddings into $r$ principal components, computes candidate-to-interest scores in compact space, and applies softmax over those scores. Its bilinear reduction is exact on the rank-$r$ reconstruction, while the normalized output approximates token-level softmax with an explicitly bounded residual. The resulting computation reduces the cost from $\mathcal{O}(N^2d)$ to $\mathcal{O}(Ndr)$, and supports 12,000 behaviors and 3,000 candidates per request without filtering. SOLAR achieves the best among compared methods on RecFlow and MIND, delivers 0.8531 AUC at approximately $19\,\mathrm{ms}$ 95th-percentile latency in an industrial evaluation, and yields business gains with a $0.68\%$ relative lift in Video Views in the real-world online A/B test. Following this evaluation, SOLAR has been fully deployed in Kuaishou's production recommendation system.

\end{abstractbox}
\end{abstract}

\input{body/intro}
\input{body/background}
\input{body/pre}
\input{body/methodology}
\input{body/results}
\FloatBarrier
\input{body/conclusion}

\bibliographystyle{plainnat}
\begingroup
\setlength{\bibsep}{3pt}
\setlength{\parskip}{0pt}
\bibliography{references}
\endgroup

\clearpage
\appendix
\startcontents[appendices]
\begingroup
\setlength{\fboxsep}{8pt}
\noindent\colorbox{gray!8}{%
\begin{minipage}{\dimexpr\linewidth-2\fboxsep\relax}
\section*{Appendix Contents}
{
\printcontents[appendices]{}{1}{\setcounter{tocdepth}{2}}
}
\end{minipage}%
}
\endgroup
\input{appdix/Appdix-SVD-Details}
\input{appdix/Appdix-Theory-Detail}
\stopcontents[appendices]

\end{document}

%% file: body/intro.tex
\section{Introduction}
\label{sec:intro}

Large-scale recommendation depends on recovering durable and evolving preferences from long behavior histories, because an infrequent past action can remain decisive when recent evidence is sparse. Transformer attention provides a natural mechanism for this task: it compares a candidate query with historical keys, uses softmax to normalize their relevance, and aggregates the corresponding values \citep{vaswani2023attentionneed}; target-attention rankers adapt this mechanism to retrieve different evidence for different items \citep{zhou2018deepnetworkclickthroughrate}. Extending the same operation from a single candidate and short history window to $N_C$ candidates and $N_L$ historical events, however, requires dense candidate-to-history interactions costing $\mathcal{O}(N_C N_L d)$, placing representational fidelity and serving efficiency in direct tension.

Existing lifelong recommenders make this regime tractable by shortening the effective history before final scoring. SIM retrieves behaviors relevant to a target item through a search stage \citep{qi2020searchbasedusermodelinglifelong}. TWIN narrows the history with a consistency-preserved two-stage design \citep{chang2023twintwostagenetworklifelong}, and TWIN V2 further compresses life-cycle behavior via hierarchical clustering \citep{Si_2024}. These systems provide scalability. Events filtered or aggregated before the final scorer cannot contribute weak or distributed evidence later, coupling prediction quality to the retrieval or compression policy.

Generic efficient-attention methods reduce the computational bottleneck through several mechanisms. Longformer restricts attention to sparse local patterns \citep{beltagy2020longformerlongdocumenttransformer}, Reformer uses locality-sensitive hashing to reduce the interaction set \citep{kitaev2020reformerefficienttransformer}, and Linformer projects the sequence into a shorter dimension \citep{wang2020linformerselfattentionlinearcomplexity}. Linear attention enables associative reordering through kernel feature maps \citep{katharopoulos2020transformersrnnsfastautoregressive}. Their efficiency comes from restricting pairwise evidence or changing the normalization that governs competition among interactions. This distinction matters for ranking because distribution-shifting linearization can yield overly smooth attention and systematically underweight high-magnitude keys \citep{fan2025rectifyingmagnitudeneglectlinear}.

Set-wise recommendation exposes a different limitation. IFA allows a candidate set to interact with a lifelong behavior sequence and removes the independent point-wise assumption \citep{yu2024ifainteractionfidelityattention}. Full candidate-to-history interaction then intensifies the product-scale cost. Retrieval-based lifelong modeling, generic efficient attention, and set-wise scoring each solve part of the problem; a unified model must retain complete histories, candidate-set conditioning, normalized competitive weighting, and predictable serving cost.

\begin{wrapfigure}{r}{0.57\linewidth}
    \centering
    \includegraphics[width=\linewidth]{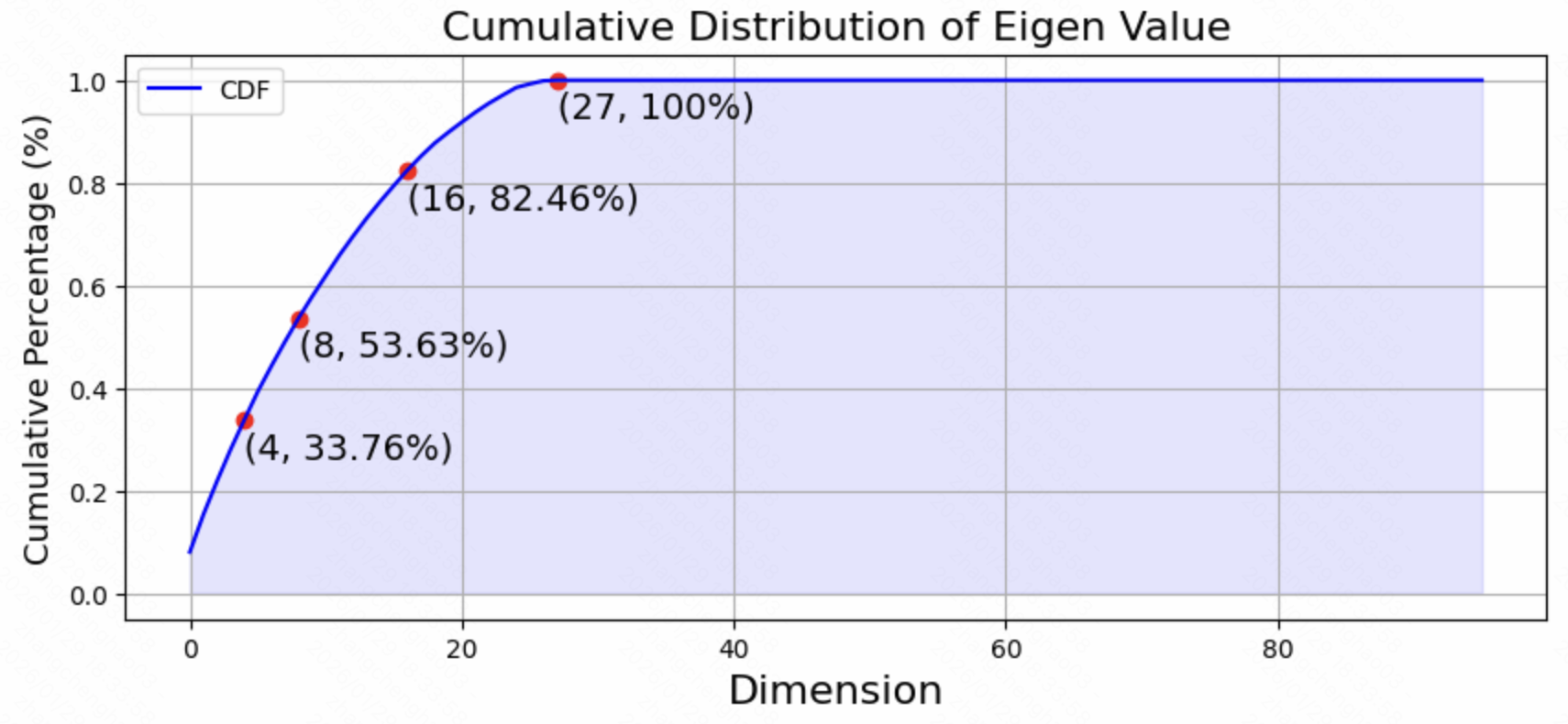}
    \caption{Low-rank nature of user sequence representations. This figure shows the cumulative distribution of eigenvalues from SVD decomposition. Rank 27 captures nearly all of the observed spectral energy.}
    \label{fig:low-rank-analysis}
\end{wrapfigure}
The high redundancy of user representations suggests another route. Low-rank factorization has long captured shared preference structure in recommendation \citep{5197422}, and recent scaling analyses show that learned recommendation embeddings increasingly concentrate in low-dimensional subspaces \citep{guo2024embeddingcollapsescalingrecommendation}. Figure~\ref{fig:low-rank-analysis} exhibits the same concentration in lifelong behavior representations. Low dimensionality alone leaves open whether compressed attention preserves token-level normalized competition and how its approximation error affects ranking quality. The missing connection among data geometry, attention fidelity, and end-to-end systems efficiency defines the gap we study.

This work asks whether the redundant structure of lifelong histories can make set-conditioned interaction practical without a prespecified retrieval filter. We answer with \textbf{SOLAR}, \underline{\textbf{S}}VD-\underline{\textbf{O}}ptimized \underline{\textbf{L}}ifelong \underline{\textbf{A}}ttention for \underline{\textbf{R}}ecommendation, a framework that couples a compact history representation with candidate-set contextualization. SVD-Attention removes the history-length-dependent left singular factor, scores each candidate against $r$ principal interests, and applies softmax over these compact scores. The bilinear reduction is exact on the rank-$r$ history reconstruction; the normalized output is an approximation to token-level softmax whose residual is bounded in the appendix.

\paragraph{Contributions.}
{
\begin{itemize}
\item We connect the low-rank geometry of lifelong behavior representations to a compact softmax attention operator. The rank-$r$ bilinear reduction is exact, while explicit bounds characterize history truncation, softmax approximation, and randomized-solver error.

\item We introduce SVD-Attention and build SOLAR as a set-aware lifelong ranker that applies softmax over compressed interests while processing 12,000 historical behaviors and 3,000 candidates per request without filtering.

\item Across RecFlow, MIND, and industrial evaluation, SOLAR achieves the strongest overall ranking results in our comparison. In the industrial systems benchmark, it remains within $0.0002$ AUC of full softmax while reducing p95 latency from approximately $76$ to $19$\,ms; a seven-day online A/B test further yields a significant $0.68\%$ relative lift in Video Views. Following this evaluation, SOLAR has been fully deployed in Kuaishou's production recommendation system.
\end{itemize}
}

%% file: body/background.tex
\section{Related Work}
\label{sec:relwork}

\subsection{Attention and Efficient Attention}

Softmax attention enables global interactions in Transformers \citep{vaswani2023attentionneed}. Its quadratic cost has motivated efficiency strategies. Longformer and Reformer restrict interactions through locality or hashing \citep{beltagy2020longformerlongdocumenttransformer,kitaev2020reformerefficienttransformer}; Set Transformer and Linformer use inducing-point or projection bottlenecks \citep{lee2019settransformerframeworkattentionbased,wang2020linformerselfattentionlinearcomplexity}; and kernelized linear attention replaces softmax with feature maps to permit associative computation \citep{katharopoulos2020transformersrnnsfastautoregressive}. These methods gain scalability by dropping or aggregating interactions, or by changing how evidence competes. The latter can alter attention sharpness and relative weighting \citep{fan2025rectifyingmagnitudeneglectlinear}.

Low-rank methods exploit redundancy in user histories. LREA compresses long-term behavior modeling through decomposition and matrix absorption \citep{Song_2025}. Our work targets normalized attention: SVD-Attention evaluates candidate-to-interest scores in compact rank-$r$ space, applies softmax to those scores, and bounds its residual from token-level softmax.

\subsection{Sequential Modeling in Recommendation Systems}

Target-attention rankers let each candidate query a behavior sequence, as in DIN and DIEN \citep{zhou2018deepnetworkclickthroughrate,zhou2018deepevolutionnetworkclickthrough}; SASRec applies Transformer-style self-attention to sequential recommendation \citep{kang2018selfattentivesequentialrecommendation}. At ranking scale, however, applying target attention to every item incurs interactions proportional to the product of candidate-set size and history length, encouraging systems to cap or filter the effective history.

SIM retrieves relevant behaviors through indexing \citep{qi2020searchbasedusermodelinglifelong}; TWIN uses a consistency-preserved two-stage model for $10^4$--$10^5$ events \citep{chang2023twintwostagenetworklifelong}; and TWIN V2 hierarchically compresses up to $10^6$ events \citep{Si_2024}. ETA-Net avoids explicit retrieval through hashing-based target attention \citep{chen2022efficientlongsequentialuser}. Although these systems establish the value of long histories, their final scores remain contingent on a selection, aggregation, or approximate-attention policy.

Candidate-set models address the independence assumption of point-wise scoring. IFA cross-attends the full candidate set to the lifelong history without filtering \citep{yu2024ifainteractionfidelityattention}, matching the set-conditioned ranking problem at product-scale attention cost. LONGER and recent scaling or generative approaches further emphasize long-sequence capacity \citep{chai2025longerscalinglongsequence,zhang2024wukongscalinglawlargescale,zhai2024actionsspeaklouderwords}. These directions motivate efficient set-wise interaction over long histories that preserves both evidence and normalized competition.

%% file: body/pre.tex
%%%%%%%%%%%%%%%%%%%%%%%%%%%%%%%%%%%%%%%%%%%%%%%%%%%%%%%%%%%%%%%%%%%%%%%%%%%%%%%%

\section{Preliminaries}

{
For a user/context $u\in\mathcal U$ with candidate set
$X_u=\{x_{u,1},\ldots,x_{u,m}\}\subset\mathcal X$, let
$\mathbf y_u\in\{0,1\}^m$ denote relevance labels and
$(u,X_u,\mathbf y_u)$ be drawn from a distribution over
$\mathcal U\times\mathcal X^m\times\{0,1\}^m$. A point-wise scorer evaluates
each item independently, whereas a set-wise scorer conditions on the full set:
\begin{equation}
f_{\mathrm{point}}(u,x_{u,i})\in\mathbb R,
\qquad
f_{\mathrm{set}}(u,x_{u,i},X_u)\in\mathbb R.
\end{equation}
Point-wise scoring is efficient and excludes effects from co-displayed items.
Set-wise scoring captures this dependence at additional computational cost.

\begin{definition}[Bipartite Ranking Risk]
\label{def:ranking_error}
For relevant indices $\mathcal P=\{i:y_i=1\}$, irrelevant indices
$\mathcal N=\{j:y_j=0\}$, and scores $s_i$, the ranking risk is
\begin{equation}
\mathcal R(f)=\mathbb E_{\mathcal D}\!\left[
\frac{1}{|\mathcal P||\mathcal N|}
\sum_{i\in\mathcal P}\sum_{j\in\mathcal N}
\mathbb I(s_j\ge s_i)\right],
\end{equation}
with the inner term defined as zero when either set is empty.
\end{definition}
Standard bipartite-ranking theory~\citep{menon2016bipartite} shows that the
posterior score
$\eta^\star(u,x_{u,i},X_u)=\mathbb P(y_i=1\mid u,x_{u,i},X_u)$
induces a Bayes-optimal ordering and minimizes $\mathcal R(f)$.

For lifelong recommendation, let the history and candidate matrices be
\begin{equation}
H=[h_1^\top;\ldots;h_{N_L}^\top]\in\mathbb R^{N_L\times d},
\qquad
C=[c_1^\top;\ldots;c_{N_C}^\top]\in\mathbb R^{N_C\times d},
\end{equation}
and define $Q=CW_Q$, $K=HW_K$, and $V=HW_V$, with
$W_Q,W_K,W_V\in\mathbb R^{d\times d}$. Token-level softmax cross-attention is
\begin{equation}
\operatorname{Attn}_{\mathrm{sm}}(Q,K,V)
=\operatorname{softmax}_{N_L}\!\left(\frac{QK^\top}{\sqrt d}\right)V,
\end{equation}
whose two dense multiplications cost $\mathcal O(N_CN_Ld)$. Linear attention
reorders the bilinear products as
\begin{equation}
\operatorname{Attn}_{\mathrm{lin}}(Q,K,V)=Q(K^\top V),
\end{equation}
reducing the cost to $\mathcal O((N_L+N_C)d^2)$ at the expense of token-level
softmax normalization and its row-stochastic weights.

Figure~\ref{fig:attention_complexity} summarizes the three computational paths
and their asymptotic costs. Section~\ref{subsec:svd_attention} develops the
low-rank formulation used by SVD-Attention.

\begin{figure}[!htbp]
    \centering
    \includegraphics[width=\linewidth]{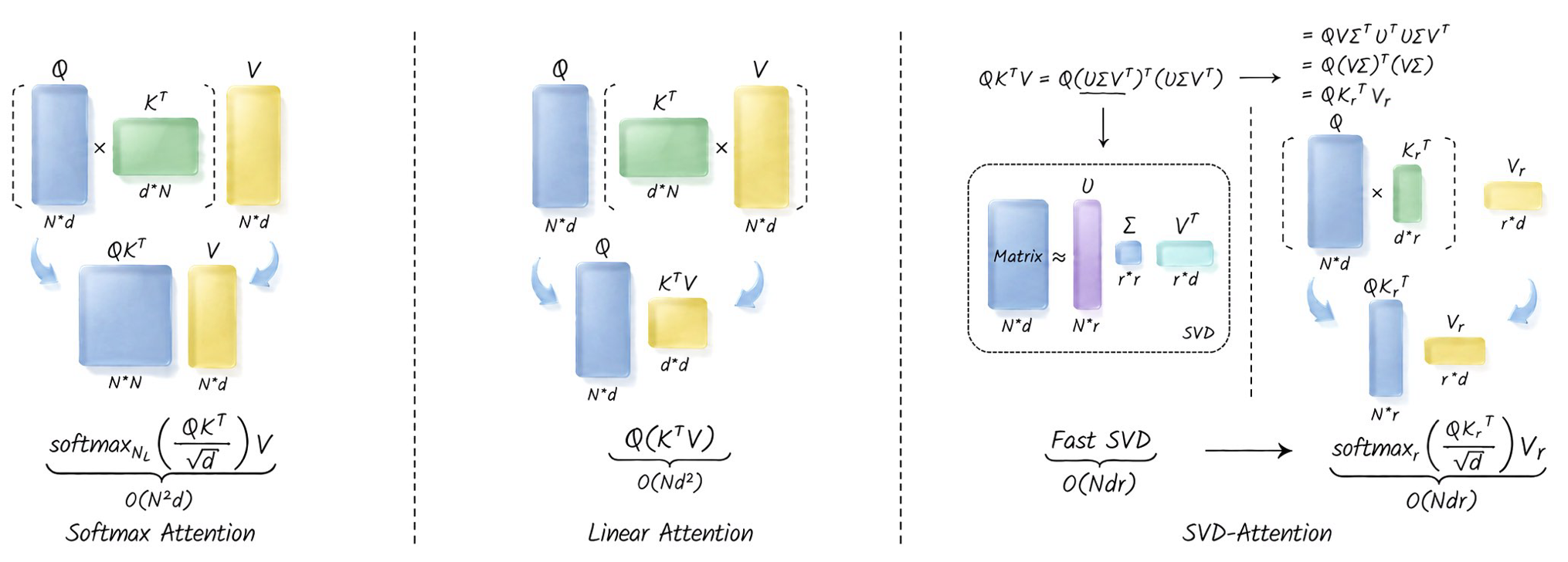}
    \caption{\textbf{Attention complexity.}
    For candidate-to-history cross-attention, softmax, linear, and SVD-Attention
    cost $\mathcal O(N_CN_Ld)$, $\mathcal O((N_L+N_C)d^2)$, and
    $\mathcal O((N_L+N_C)dr)$, respectively; these become
    $\mathcal O(N^2d)$, $\mathcal O(Nd^2)$, and $\mathcal O(Ndr)$ when
    $N_L,N_C=\Theta(N)$. The diagram shows the dominant matrix
    multiplications; SVD-Attention additionally applies softmax to $r$
    candidate--interest scores at cost $\mathcal O(N_Cr)$, and the SVD cost
    includes randomized factorization.}
    \label{fig:attention_complexity}
\end{figure}
}

%% file: body/methodology.tex
%%%%%%%%%%%%%%%%%%%%%%%%%%%%%%%%%%%%%%%%%%%%%%%%%%%%%%%%%%%%%%%%%%%%%%%%%%%%%%%%

\section{Methodology}
\label{sec:methodology}

{
SOLAR addresses two coupled requirements of lifelong recommendation. Each candidate must retrieve useful evidence from thousands of historical behaviors without forming a dense candidate--history attention map, and its score must reflect the other candidates in the same request. SVD-Attention compresses the history, scores each candidate against the resulting principal interests, and normalizes those scores with softmax; a candidate-set branch supplies the context required for set-conditioned ranking.

\subsection{SVD-Attention over Principal Interests}
\label{subsec:svd_attention}

Recall the history matrix $H\in\mathbb{R}^{N_L\times d}$ and candidate matrix $C\in\mathbb{R}^{N_C\times d}$ from Section~3. Standard cross-attention uses $C$ to query all $N_L$ historical tokens at a cost of $\mathcal{O}(N_CN_Ld)$. The empirical spectrum in Figure~\ref{fig:low-rank-analysis} suggests that most variation in $H$ is concentrated in far fewer directions. We approximate $H$ by its rank-$r$ reconstruction
\begin{equation}
    H_r=U\Sigma V^\top\approx H,
    \qquad
    U\in\mathbb{R}^{N_L\times r},\ 
    \Sigma\in\mathbb{R}^{r\times r},\ 
    V\in\mathbb{R}^{d\times r},
    \label{eq:history_svd}
\end{equation}
where $r\ll\min(N_L,d)$ and $U^\top U=I_r$. The spectral tail makes the first source of approximation explicit:
\begin{equation}
    \|H-H_r\|_F^2=\sum_{j>r}\sigma_j^2.
    \label{eq:svd_truncation_error}
\end{equation}

\paragraph{An exact low-rank starting point.}
Let $\mathrm{Query}=CW_Q$ and define compact keys and values
\begin{equation}
\mathrm{Key}_r=\Sigma V^\top W_K,
\qquad
\mathrm{Value}_r=\Sigma V^\top W_V.
\label{eq:compact_kv}
\end{equation}
The keys and values induced by $H_r$ are $U\mathrm{Key}_r$ and $U\mathrm{Value}_r$. Consequently, the unnormalized bilinear attention core on $H_r$ satisfies
\begin{equation}
\begin{aligned}
\operatorname{Attn}_{\mathrm{bilinear}}^{(r)}
&=\mathrm{Query}\,(U\mathrm{Key}_r)^\top U\mathrm{Value}_r\\
&=\mathrm{Query}\,\mathrm{Key}_r^\top(U^\top U)\mathrm{Value}_r\\
&=\mathrm{Query}\,\mathrm{Key}_r^\top\mathrm{Value}_r.
\end{aligned}
\label{eq:svd_bilinear_core}
\end{equation}
This identity is exact for the bilinear operator evaluated on $H_r$ and provides the computational basis shown in Figure~\ref{fig:attention_complexity}. The relation $U^\top U=I_r$ removes the history-length-dependent factor $U$. Under an exact rank-$r$ factorization, Eq.~\eqref{eq:history_svd} contributes the only approximation to the bilinear scores relative to $H$; randomized factorization adds the solver error evaluated in Section~\ref{sec:approximation_fidelity}. Because softmax is nonlinear, Eq.~\eqref{eq:svd_bilinear_core} does not imply an exact reformulation of token-level softmax attention.

\paragraph{Softmax over principal interests.}
SVD-Attention forms compact scores and applies standard row-wise softmax before aggregating the compact values:
\begin{equation}
S_r=\frac{\mathrm{Query}\,\mathrm{Key}_r^\top}{\sqrt d},
\qquad
\operatorname{Attn}_{\mathrm{SVD}}(C,H)
=\operatorname{softmax}_{r}(S_r)\mathrm{Value}_r.
\label{eq:svd_softmax_attention}
\end{equation}

\begin{wraptable}{r}{0.57\linewidth}
    \centering
    
    \caption{Complexity of the dominant attention paths in Figure~\ref{fig:attention_complexity}. The third column reports the asymptotic cost under equal-size scaling $N_L,N_C=\Theta(N)$; softmax over $r$ scores does not change the SVD-Attention order.}
    \label{tab:complexity_comparison}
    \resizebox{\linewidth}{!}{
    \begin{tabular}{lcc}
        \toprule
        \textbf{Method} & \textbf{Cross-attention cost} & \textbf{Cost for $N_L,N_C=\Theta(N)$} \\
        \midrule
        Softmax Attention & $\mathcal{O}(N_CN_Ld)$ & $\mathcal{O}(N^2d)$ \\
        Linear Attention & $\mathcal{O}((N_L+N_C)d^2)$ & $\mathcal{O}(Nd^2)$ \\
        \midrule
        \textbf{SVD-Attention (ours)} & $\mathcal{O}((N_L+N_C)dr)$ & $\mathcal{O}(Ndr)$ \\
        \bottomrule
    \end{tabular}
    }
\end{wraptable}

Equation~\eqref{eq:svd_softmax_attention} is the SVD-Attention operator used by SOLAR throughout the paper. It normalizes over $r$ principal interests rather than $N_L$ historical tokens, so its output approximates token-level softmax attention even when the rank-$r$ factorization is exact. Appendix~\ref{app:svd_softmax_analysis} states the conditions under which the two normalized outputs approach each other and proves the corresponding residual bound. Results labeled ``without softmax'' retain only the bilinear computation in Eq.~\eqref{eq:svd_bilinear_core}.

\paragraph{Computation.}
Randomized rank-$r$ decomposition costs $\mathcal{O}(N_Ldr)$, and the two compact candidate-side multiplications cost $\mathcal{O}(N_Cdr)$. The total forward complexity of SVD-Attention is therefore $\mathcal{O}((N_L+N_C)dr)$ without constructing an $N_C\times N_L$ interaction matrix. When $N_L,N_C=\Theta(N)$, this complexity reduces to $\mathcal{O}(Ndr)$, as shown in Figure~\ref{fig:attention_complexity}. As in standard attention analyses, the comparison counts the attention kernel after the learned projections; the compact factor projections add $\mathcal{O}(rd^2)$, which is lower order in the target regime. Row-wise softmax over $r$ components adds only $\mathcal{O}(N_Cr)$ and does not change the leading order.

\begin{wrapfloat}{algorithm}{r}{0.57\linewidth}

\captionsetup{labelfont=bf,labelsep=space}
\hrule
\caption{Randomized SVD with Power Iteration}
\label{alg:randsvd}
\small
\begin{algorithmic}
\REQUIRE History matrix $H\in\mathbb{R}^{N_L\times d}$; rank $r$; iterations $n_{\mathrm{iter}}$.
\ENSURE Singular values $\Sigma\in\mathbb{R}^{r\times r}$ and right singular vectors $V\in\mathbb{R}^{d\times r}$.
\STATE Initialize $\Omega\in\mathbb{R}^{d\times r}$ with $\Omega_{ij}\sim\mathcal{N}(0,1)$.
\FOR{$i=1$ to $n_{\mathrm{iter}}$}
    \STATE $\Omega\leftarrow H^\top(H\Omega)$.
\ENDFOR
\STATE Compute an orthonormal basis $Q$ of $H\Omega$ by QR decomposition.
\STATE Compute $Q^\top H=U_S\Sigma V^\top$.
\STATE \textbf{Return} $\Sigma,V$.
\end{algorithmic}
\hrule
\end{wrapfloat}

\paragraph{Efficient and differentiable factorization.}
We avoid a full decomposition by using randomized SVD with power iteration~\citep{halko2011finding}. Algorithm~\ref{alg:randsvd} returns only the singular values and right singular vectors required by Eq.~\eqref{eq:compact_kv}; $U$ need not be materialized by the attention operator.

The retained factors participate in end-to-end training. Following matrix backpropagation~\citep{ionescu2015matrix}, gradients are propagated through $\Sigma$ and $V$ within the truncated subspace; the unused $U$ branch has zero upstream gradient. Appendix~\ref{app:svd_derivation} gives the complete gradient expression, its derivation, and the associated spectral-gap conditions. The resulting training path is well defined and retains the compact forward operator.

\subsection{SOLAR for Set-Conditioned Lifelong Ranking}
\label{set-wise}

\begin{wrapfigure}[10]{r}{0.52\linewidth}
    \centering
    \includegraphics[width=\linewidth]{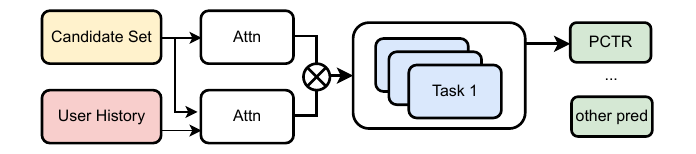}
    \captionsetup{skip=2pt}
    
    \caption{SOLAR combines SVD-Attention over the lifelong user history with candidate-set interaction before the task-specific prediction heads.}
    \label{fig:solar_architecture}
\end{wrapfigure}

SVD-Attention provides efficient candidate-to-history interaction. Set-conditioned ranking additionally requires candidate context, so SOLAR combines SVD-Attention with a candidate-set branch, as shown in Figure~\ref{fig:solar_architecture}. The history branch applies Eq.~\eqref{eq:svd_softmax_attention} with each candidate as a query to produce a candidate-specific lifelong-interest representation $Z_H$. The candidate branch contextualizes each item with the other candidates in the current request, producing $Z_C$. Their fused representation enters the shared and task-specific prediction layers, making each score a function of the user history, focal candidate, and surrounding candidate set:
\begin{equation}
Z_H=\operatorname{Attn}_{\mathrm{SVD}}(C,H),
\qquad
Z_C=\operatorname{SetAttn}(C),
\qquad
s_i=g(c_i,(Z_H)_{i:},(Z_C)_{i:}).
\label{eq:solar_dataflow}
\end{equation}

The candidate-set branch is motivated by a simple limitation of independent scoring. Suppose two items reverse their relative relevance across two candidate sets:
\begin{equation}
\eta^\star(x_i,X^{(1)})>\eta^\star(x_j,X^{(1)}),
\qquad
\eta^\star(x_i,X^{(2)})<\eta^\star(x_j,X^{(2)}).
\label{eq:contextual_flip_main}
\end{equation}

\begin{proposition}[Necessity of Candidate-Set Conditioning]
\label{prop:contextual_conditioning}
Under Eq.~\eqref{eq:contextual_flip_main}, no context-independent point-wise scorer $f_{\mathrm{point}}(x)$ can represent both orderings, whereas a realizable set-wise scorer $f_{\mathrm{set}}(x,X)=\eta^\star(x,X)$ can.
\end{proposition}
The difference $f_{\mathrm{point}}(x_i)-f_{\mathrm{point}}(x_j)$ has a fixed sign for every $X$, whereas the target difference changes sign. Appendix~\ref{app:setwise_formal} gives the corresponding risk-based statement and proof. This result motivates supplying the candidate set to the scorer, independently of the low-rank derivation and softmax approximation in SVD-Attention.

Candidate interaction is also useful when items within a request share strong components. The block-dependent analysis in Appendix~\ref{app:setwise_formal} introduces the factor
\begin{equation}
\mathcal C(m,\rho)=\sqrt{1+(m-1)\rho},
\label{eq:correlation_factor_main}
\end{equation}
where $m$ is the candidate-set size and $\rho$ is the assumed average within-request representation correlation. This factor grows with $\rho$, motivating an interaction block that can emphasize set-dependent differences alongside shared features. Because $\rho$ is not measured directly, the bound provides architectural motivation but does not establish that SOLAR learns a smaller correlation or a tighter empirical generalization gap.

The contextual-flip result motivates candidate-set conditioning, and the correlation analysis motivates explicit candidate interaction. The spectral residual in Eq.~\eqref{eq:svd_truncation_error} characterizes rank truncation, while Appendix~\ref{app:svd_softmax_analysis} bounds the normalized SVD-Attention output relative to token-level softmax. These components yield a set-conditioned model whose dominant candidate-to-history computation is $\mathcal{O}((N_L+N_C)dr)$, or $\mathcal{O}(Ndr)$ when $N_L,N_C=\Theta(N)$.
}

%% file: body/results.tex
%%%%%%%%%%%%%%%%%%%%%%%%%%%%%%%%%%%%%%%%%%%%%%%%%%%%%%%%%%%%%%%%%%%%%%%%%%%%%%%%

\section{Evaluation}
\label{sec:eval}

We evaluate SOLAR on RecFlow, MIND, and a large-scale industrial benchmark. Dataset-specific settings are provided in Appendix Table~\ref{tab:eval_settings}; the industrial setting contains 12,000 historical behaviors and 3,000 candidates per request. The evaluation addresses four questions: \textbf{Q1}, whether SOLAR improves ranking quality across offline benchmarks and production traffic; \textbf{Q2}, how much accuracy SVD-Attention retains and whether it provides a favorable quality--latency trade-off; \textbf{Q3}, how history length, retained rank, softmax normalization, and the two modeling branches affect performance; and \textbf{Q4}, whether SOLAR satisfies production resource constraints and improves online business metric.

\subsection{Experimental Setup}

\textbf{Benchmarks.} RecFlow \citep{liu2025recflow} and MIND \citep{wu-etal-2020-mind} serve as offline benchmarks. Each instance contains a length-50 history and approximately 50--130 candidates, constructed as described in Appendix~\ref{datasets}. The industrial evaluation uses production traffic. We measure industrial AUC and UAUC over a 24-hour window after convergence in streaming training, and evaluate business impact separately in the seven-day online A/B test in Section~\ref{sec:online_ab}.

\textbf{Metrics.} We report AUC and Logloss for RecFlow and MIND because the task is calibrated click-probability ranking over sampled candidates; these metrics match both the training objective and the evaluation convention of industrial sequence models such as DIN, SIM, TWIN, and LONGER. Recall and NDCG are not included in this sampled-candidate protocol. The industrial benchmark reports AUC and UAUC, while the deployment study additionally measures latency, throughput, resource use, service availability, and Video Views.

\textbf{Baseline protocol.} Two-stage baselines such as SIM \citep{qi2020searchbasedusermodelinglifelong} explicitly filter the history before scoring. On RecFlow and MIND, they select 20 behaviors from the length-50 history; in the industrial experiment, they select 300 behaviors from the length-12,000 history. Other baselines use the same input setting as SOLAR.

\subsection{Overall Predictive Effectiveness (Q1)}
\label{sec:main_results}

Table~\ref{main-result} compares SOLAR with representative sequential and lifelong recommendation models. SOLAR achieves the best AUC and Logloss on both offline benchmarks and the best AUC and UAUC on industrial traffic. Relative to the strongest set-aware baseline, IFA \citep{yu2024ifainteractionfidelityattention}, SOLAR improves AUC by $0.0043$ on RecFlow, $0.0016$ on MIND, and $0.0013$ in the industrial setting; industrial UAUC also improves by $0.0013$. The gains are consistent across public and industrial data and across both ranking and calibration metrics.

\begin{table}[H]
\centering

\caption{Main results on RecFlow, MIND, and the industrial online setting. A 0.001-level improvement is considered significant for industrial CTR prediction, as supported by prior experiments \citep{ESMM, Si_2024, WangSCJLHC21}.}
\label{main-result}
\resizebox{1.0\textwidth}{!}{
\begin{tabular}{lcccccccc>{\columncolor{blue!5}}c}
\toprule
Dataset & Metric
& \shortstack{DIN (Recent 50) \\ \citep{zhou2018deepnetworkclickthroughrate}}
& \shortstack{SIM \\ \citep{qi2020searchbasedusermodelinglifelong}}
& \shortstack{LREA \\ \citep{Song_2025}}
& \shortstack{LONGER \\ \citep{chai2025longerscalinglongsequence}}
& \shortstack{TWIN \\ \citep{chang2023twintwostagenetworklifelong}}
& \shortstack{TWINv2 \\ \citep{Si_2024}}
& \shortstack{IFA \\ \citep{yu2024ifainteractionfidelityattention}}
& \shortstack{SOLAR \\ ours} \\
\midrule
\multirow{2}{*}{RecFlow}
  & AUC $\uparrow$       & 0.6048 & 0.6113 & 0.5695 & 0.6676 & 0.6057 & 0.6178 & 0.6769 & 0.6812 \\
  & Logloss $\downarrow$ & 0.0689 & 0.0673 & 0.1475 & 0.0638 & 0.0687 & 0.0691 & 0.0615 & 0.0610 \\
\midrule
\multirow{2}{*}{MIND}
  & AUC $\uparrow$       & 0.5909 & 0.6193 & 0.6211 & 0.6355 & 0.6480 & 0.6377 & 0.6697 & 0.6713 \\
  & Logloss $\downarrow$ & 0.1324 & 0.1156 & 0.1065 & 0.1137 & 0.1075 & 0.1164 & 0.1060 & 0.1052 \\
\midrule
\multirow{2}{*}{Industrial}
  & AUC ($mean^{\pm std}$) $\uparrow$ & $0.8417^{\pm 0.00487}$ & $0.8458^{\pm 0.00512}$ & $0.8479^{\pm 0.00523}$ & $0.8511^{\pm 0.00413}$ & $0.8507^{\pm 0.00508}$ & $0.8492^{\pm 0.00492}$ & $0.8518^{\pm 0.00569}$ & $0.8531^{\pm 0.00589}$ \\
  & UAUC ($mean^{\pm std}$) $\uparrow$ & $0.8398^{\pm 0.00475}$ & $0.8439^{\pm 0.00515}$ & $0.8460^{\pm 0.00519}$ & $0.8486^{\pm 0.00405}$ & $0.8488^{\pm 0.00495}$ & $0.8473^{\pm 0.00487}$ & $0.8489^{\pm 0.00504}$ & $0.8502^{\pm 0.00530}$ \\
\bottomrule
\end{tabular}
}
\end{table}
\FloatBarrier

\subsection{Approximation Fidelity and Efficiency (Q2)}

We examine approximation fidelity and efficiency at three levels: a controlled approximation study, an end-to-end comparison under a shared production configuration, and a single-thread scaling benchmark. The controlled latency in Figure~\ref{fig:approximation_fidelity} isolates the attention computation and is therefore not directly comparable to the end-to-end serving latency in Table~\ref{tab:efficient_attention}.

\paragraph{Approximation fidelity.}
\label{sec:approximation_fidelity}

Figure~\ref{fig:approximation_fidelity} decomposes the path from full token softmax to the deployed randomized SVD-Attention operator. It reports relative output deviation from full token softmax ($E_{\mathrm{total}}$), the corresponding AUC decrease, and controlled latency. Appendix Table~\ref{tab:approximation_fidelity_full} provides the complete error decomposition; $\sim$ marks approximate measurements.

\begin{figure}[H]
\centering
\includegraphics[width=0.98\linewidth]{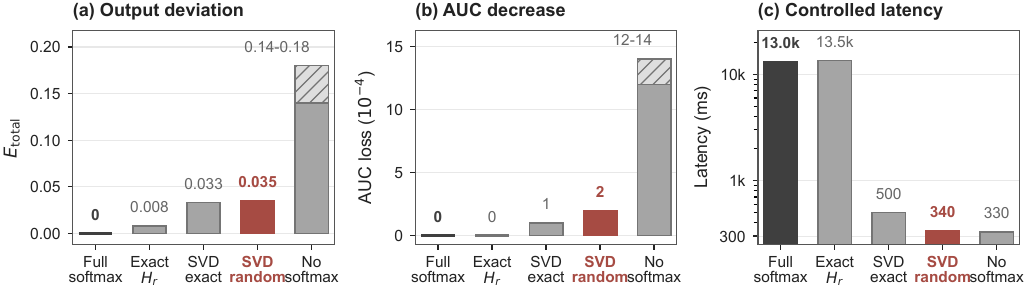}
\caption{\textbf{Approximation quality and controlled latency.} Panel (a) measures output deviation from full token softmax; panel (b) reports the AUC decrease from its $0.8539$ reference; and panel (c) reports latency on a logarithmic scale. Randomized SVD-Attention remains close to the reference and operates near the no-softmax ablation. Hatched bar segments denote measured intervals; latency values are approximate.}
\label{fig:approximation_fidelity}
\end{figure}

The exact rank-$r$ reconstruction preserves full-softmax AUC with only $0.008$ output deviation, indicating limited predictive loss from truncation. Applying softmax to the compact principal-interest scores increases the deviation to $0.033$ and changes AUC by only $0.0001$. Randomized SVD adds merely $0.002$--$0.004$ solver error and reduces latency from approximately $500$ to $340$\,ms. The no-softmax ablation is marginally faster and less accurate, supporting the normalized SVD-Attention operator used by SOLAR. Rank-wise results appear in Appendix Table~\ref{tab:rank_fidelity}.

\paragraph{System-level comparison and scaling.}
\label{sec:forward_efficiency}

The system-level comparison uses the industrial setting from Appendix Table~\ref{tab:eval_settings}, with a shared model framework and production CPU configuration. Table~\ref{tab:efficient_attention} reports ranking quality, p50/p95 latency, peak cores usage of CPU, and throughput. In the last column, $\checkmark$ marks p95 latency at or below the 25\,ms target, and $\times$ marks higher latency. The resource consumption for online serving not only includes the forward propagation
time but also the actual overhead of fetching embeddings and other related operations.

\begin{table}[H]
\centering
\caption{Quality and system efficiency of attention operators in the shared industrial setting.}
\label{tab:efficient_attention}
\begingroup
\scriptsize
\setlength{\tabcolsep}{2.6pt}
\renewcommand{\arraystretch}{1.12}
\begin{tabular*}{\textwidth}{@{\extracolsep{\fill}}lcccccc@{}}
\toprule
\multirow{2}{*}{\textbf{Method}} &
\multirow{2}{*}{\shortstack{\textbf{Selected}\\\textbf{size}}} &
\multicolumn{1}{c}{\textbf{Quality}} &
\multicolumn{4}{c}{\textbf{System efficiency}} \\
\cmidrule(lr){3-3}\cmidrule(l){4-7}
& & \textbf{AUC}$\uparrow$ &
\shortstack{\textbf{Latency (ms)}\\\textbf{p50/p95}$\downarrow$} &
\shortstack{\textbf{Cores Usage}\\\textbf{$\Delta$}$\downarrow$} &
\shortstack{\textbf{Throughput}\\\textbf{(QPS)}$\uparrow$} &
\shortstack{\textbf{p95 target}\\$\boldsymbol{\leq25}$\textbf{ ms}} \\
\midrule
Full softmax attention & $N_L$ & $0.8533$ & $48/76$ & base & $\times 1.0$ & $\times$ \\
Linear attention \citep{katharopoulos2020transformersrnnsfastautoregressive} & $d=128$ & $0.8518$ & $18/24$ & -38.09\% & $\times 2.29$ & $\checkmark$ \\
Linformer \citep{wang2020linformerselfattentionlinearcomplexity} & $k=64$ & $0.8510$ & $17/24$ & -27.02\%  & $\times 2.63$ & $\checkmark$ \\
Performer \citep{choromanski2021rethinkingattentionperformers} & $m=64$ & $0.8507$ & $18/24.5$ & -27.02\% & $\times 2.45$ & $\checkmark$ \\
Longformer \citep{beltagy2020longformerlongdocumenttransformer} & $w=256$ & $0.8537$ & $19/27$ & -23.80\% & $\times 2.27$ & $\times$ \\
Hiformer \citep{gui2023hiformerheterogeneousfeatureinteractions} & $k=64$ & $0.8312$ & $19/27$ & -23.80\% & $\times 2.22$ & $\times$ \\
MALA \citep{fan2025rectifyingmagnitudeneglectlinear} & $m=64$ & $0.8525$ & $17/24$ & -38.09\% & $\times 2.54$ & $\checkmark$ \\
\midrule
SVD-Attention (exact SVD) & $r=27$ & $0.8532$ & - & - & - & - \\
\textbf{SVD-Attention (SOLAR)} & \textbf{$r=27$} & \textbf{0.8531} & \textbf{13/19} & \textbf{-52.38\%} & $\times 3.36$ & \textbf{$\checkmark$} \\
\bottomrule
\end{tabular*}
\endgroup
\end{table}

SOLAR stays within $0.0002$ AUC of full softmax while reducing p95 latency from $76$ to $19$\,ms and peak core usage was reduced by 52.38\% compared to full softmax attention, with the highest throughput. Linear Attention, Linformer, Performer, and MALA meet the latency target at larger predictive losses; Longformer retains competitive AUC but misses the target. 

Figure~\ref{fig:scaling_sensitivity}(a) isolates the three computational paths in Figure~\ref{fig:attention_complexity} on one CPU thread, separating algorithmic scaling from parallel scheduling. At $N=30{,}000$, SVD-Attention is approximately $79\times$ faster than Softmax Attention; its softmax over $r$ principal interests is lower order and leaves the asymptotic scaling unchanged. Overall, randomized SVD-Attention retains near-full-softmax quality, meets the production latency target, and scales gently with history length.

\subsection{Sensitivity and Ablation (Q3)}

\paragraph{History length and retained rank.}
\label{sec:length_rank_sensitivity}

We vary one factor at a time in the industrial offline pipeline. The history-length study ranges from 50 to 12,000 behaviors, and the rank study varies $r$ from 4 to 64; all remaining settings are fixed. Figure~\ref{fig:scaling_sensitivity}(b) and (c) summarize both trends alongside the latency scaling result.

\begin{figure}[H]
\centering
\includegraphics[width=\linewidth]{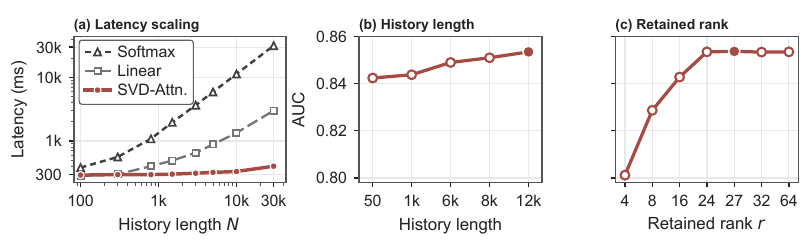}
\caption{\textbf{SVD-Attention scales gently while preserving the value of long histories at modest rank.} Panel (a) reports single-thread CPU latency with fixed candidate size $m$ and embedding dimension $d$. Panels (b) and (c) report controlled industrial offline AUC versus retained history length and SVD rank; filled markers denote the deployed settings (12,000 behaviors and $r=27$). Table~\ref{main-result} reports AUC and UAUC from production traffic.}
\label{fig:scaling_sensitivity}
\end{figure}

Increasing history length from 50 to 12,000 improves AUC from $0.8424$ to $0.8535$, confirming that the target regime benefits from evidence beyond recent behavior. Reducing $r$ from 64 to 27 causes no degradation in this controlled run; AUC drops sharply below rank 16. The useful signal grows with history length and remains concentrated in a modest principal subspace, consistent with the spectrum in Figure~\ref{fig:low-rank-analysis}.

\paragraph{Normalization and modeling components.}
\label{sec:ablation}
\label{sec:ablation_attention}

Table~\ref{tab:ablation_attention_vertical} separates the effects of softmax normalization, candidate-set contextualization, and lifelong history modeling. Table~\ref{tab:efficient_attention} provides the corresponding controlled comparison of efficient-attention operators. Resource savings are measured under 6.14\% production traffic and include embedding retrieval and other serving overheads in addition to forward propagation.

\begin{table}[htbp]
\centering
\caption{Normalization and component ablations under the shared set-aware framework. Core use is relative to full softmax; ``--'' denotes an unmeasured value.}
\label{tab:ablation_attention_vertical}
\begingroup
\footnotesize
\setlength{\tabcolsep}{2.8pt}
\renewcommand{\arraystretch}{1.12}
\begin{tabular*}{\textwidth}{@{\extracolsep{\fill}}lcccccc@{}}
\toprule
\multirow{2}{*}{\textbf{Variant}} &
\multicolumn{3}{c}{\textbf{Components}} &
\multicolumn{2}{c}{\textbf{Ranking quality}} &
\multicolumn{1}{c}{\textbf{Efficiency}} \\
\cmidrule(lr){2-4}\cmidrule(lr){5-6}\cmidrule(l){7-7}
& \shortstack{\textbf{Softmax}\\\textbf{norm.}} &
\textbf{History} &
\shortstack{\textbf{Candidate}\\\textbf{set}} &
\shortstack{\textbf{RecFlow}\\\textbf{AUC}$\uparrow$} &
\shortstack{\textbf{Industrial AUC}$\uparrow$\\$(mean^{\pm std})$} &
\shortstack{\textbf{$\Delta$ core}\\\textbf{use}$\downarrow$} \\
\midrule
Full softmax attention & $\checkmark$ & $\checkmark$ & $\checkmark$ & 0.6829 & $0.8533^{\pm 0.00451}$ & base \\
SVD bilinear (no softmax) & $\times$ & $\checkmark$ & $\checkmark$ & 0.6757 & $0.8520^{\pm 0.00656}$ & -- \\
Candidate-set branch only & $\times$ & $\times$ & $\checkmark$ & 0.6293 & $0.8012^{\pm 0.00408}$ & -- \\
User-history branch only & $\checkmark$ & $\checkmark$ & $\times$ & 0.6644 & $0.8508^{\pm 0.00496}$ & -- \\
\midrule
\textbf{SVD-Attention (SOLAR)} & $\boldsymbol{\checkmark}$ & $\boldsymbol{\checkmark}$ & $\boldsymbol{\checkmark}$ & \textbf{0.6812} & \textbf{$0.8531^{\pm 0.00589}$} & \textbf{$-52.38\%$} \\
\bottomrule
\end{tabular*}
\endgroup
\end{table}

Softmax normalization improves the unnormalized bilinear ablation by $0.0055$ AUC on RecFlow and $0.0011$ on the industrial benchmark, recovers most of the gap to full token softmax, and reduces core use by $52.38\%$. Removing history modeling causes the largest degradation. Removing candidate-set contextualization lowers full-model AUC by $0.0168$ on RecFlow and $0.0023$ industrially. The two branches are complementary: lifelong history supplies the primary evidence, and candidate-set context refines how that evidence is used for the current request.

The sensitivity and ablation results place the long-history signal near $r=27$ and show measurable gains from both softmax normalization and candidate-set context.

\subsection{Online Deployment (Q4)}
\label{sec:online_ab}

We evaluated SOLAR through an online A/B test in Kuaishou's production recommendation application, which serves more than 400 million daily active users (DAU). The control and treatment groups each received $6.14\%$ of production traffic. Compared with the production baseline, SOLAR achieved a statistically significant relative increase of $0.68\%$ in video views ($p<0.05$). Following the successful A/B test, SOLAR has been fully deployed in Kuaishou's production recommendation system.

The serving latency target is 25\,ms per request. Table~\ref{tab:efficient_attention} shows that SOLAR reaches approximately 19\,ms p95 latency; Table~\ref{tab:ablation_attention_vertical} reports a $52.38\%$ reduction in core use relative to full softmax. During peak hours from 8 PM to 10 PM, throughput is 20,000 queries per second and availability remains above 99\%. Each request processes 12,000 historical behaviors and 3,000 candidates.

These deployment results confirm that SOLAR's offline quality--efficiency trade-off translates into lower resource use and significant real-world business gains. For context, a relative increase of $0.12\%$ in video views has been considered a meaningful business gain.

%% file: body/conclusion.tex
\section{Conclusion}
\label{sec:conclusion}
Dense attention becomes costly for long sequences. SVD-Attention factorizes the history, computes candidate-to-interest scores in compact space, and applies softmax over $r$ principal interests. Its bilinear reduction is exact on the rank-$r$ reconstruction, while Appendix~\ref{app:svd_softmax_analysis} bounds the normalized approximation to token-level softmax. This design reduces candidate-to-history computation from $\mathcal{O}(N_CN_Ld)$ to $\mathcal{O}((N_L+N_C)dr)$, or $\mathcal{O}(Ndr)$ when $N_L,N_C=\Theta(N)$. Combined with candidate-set contextualization, SOLAR ranks 3,000 candidates against 12,000 historical behaviors without filtering and improves quality in the evaluated settings.

%% file: appdix/Appdix-SVD-Details.tex
\section{Reproducibility}
The code will be published upon acceptance.

\begin{table}[htbp]
\centering

\caption{Dataset-specific evaluation settings. Offline benchmarks use length-50 histories, with 120 candidates on RecFlow and 50--130 candidates on MIND. Each industrial request contains 12,000 historical behaviors and 3,000 candidates.}
\label{tab:eval_settings}
\resizebox{0.65\linewidth}{!}{
\begin{tabular}{lccc}
\toprule
\textbf{Dataset} & \textbf{Evaluation} & \textbf{History length} & \textbf{Candidate size} \\
\midrule
RecFlow & Offline & 50 & 120 \\
MIND & Offline & 50 & 50--130 \\
Industrial & Online & 12,000 & 3,000 \\
\bottomrule
\end{tabular}
}
\end{table}

We use the dataset-specific configurations in Table~\ref{tab:training_hyperparameters}. RecFlow and MIND use rank $r=8$ with 64-dimensional embeddings, whereas the industrial setting uses rank $r=27$ with 128-dimensional embeddings. All settings use four attention heads and the Adam optimizer.

\begin{table}[htbp]
\centering

\caption{Training hyperparameters used in the experiments.}
\label{tab:training_hyperparameters}
\resizebox{\linewidth}{!}{
\begin{tabular}{lcccll}
\toprule
\textbf{Dataset} & \textbf{Rank} & \textbf{Embedding dim.} & \textbf{Heads} & \textbf{Batch size} & \textbf{Optimizer and learning rate} \\
\midrule
RecFlow & 8 & 64 & 4 & 2,048 & Adam, $0.001$ \\
MIND & 8 & 64 & 4 & 2,048 & Adam, $0.001$ \\
Industrial & 27 & 128 & 4 & 4,096 & Adam, CosineAnnealingLR from $0.01$ \\
\bottomrule
\end{tabular}
}
\end{table}

\paragraph{Rank selection.} We select $r$ in two stages, analogously to selecting the number of principal components in PCA. We first use cumulative explained spectral energy to identify a reasonable range and then choose the final value by balancing predictive effectiveness and serving efficiency. The controlled sensitivity results in Figure~\ref{fig:scaling_sensitivity}(c) show the resulting performance--rank trade-off.

\begin{table*}[htbp]
\centering

\caption{Approximation fidelity at different retained ranks. The error definitions follow Section~\ref{sec:approximation_fidelity}.}
\label{tab:rank_fidelity}
\resizebox{0.88\textwidth}{!}{
\begin{tabular}{ccccccc}
\toprule
$r$ & \textbf{Spectral energy}$\uparrow$ & $E_H\downarrow$ & $E_{\mathrm{total}}\downarrow$ & $E_{\mathrm{softmax}}\downarrow$ & $E_{\mathrm{solver}}\downarrow$ & \textbf{AUC}$\uparrow$ \\
\midrule
4  & \textbf{33.76\%} & \textbf{0.814} & $0.50$ & $0.12$ & $0.002$ & \textbf{0.8012} \\
8  & \textbf{53.63\%} & \textbf{0.681} & $0.34$ & $0.09$ & $0.002$ & \textbf{0.8287} \\
16 & \textbf{82.46\%} & \textbf{0.419} & $0.15$ & $0.06$ & $0.002$ & \textbf{0.8428} \\
24 & $95$--$97\%$ & $0.17$--$0.22$ & $0.055$--$0.070$ & $0.035$--$0.045$ & $0.002$--$0.004$ & \textbf{0.8535} \\
27 & $\approx100\%$ & $0.01$--$0.03$ & $0.03$--$0.04$ & $0.025$--$0.035$ & $0.002$--$0.004$ & \textbf{0.8537} \\
32 & $\approx100\%$ & $<0.02$ & $0.03$ & $0.028$ & $0.003$ & \textbf{0.8534} \\
64 & $\approx100\%$ & $<0.01$ & $0.025$--$0.03$ & $0.025$ & $0.003$ & \textbf{0.8535} \\
\bottomrule
\end{tabular}
}
\end{table*}

The fidelity transition is sharp between ranks 16 and 24: increasing the retained spectral energy from $82.46\%$ to approximately $95$--$97\%$ reduces total error from about $0.15$ to $0.055$--$0.070$ and raises AUC from $0.8428$ to $0.8535$. Rank 27 nearly saturates the spectrum and gives the best AUC in this sweep, while the randomized solver remains a minor component of the total error.

\section{Detailed Derivation of SVD Gradients}
\label{app:svd_derivation}

This section derives the SVD gradient with respect to the input matrix $H$ for truncated SVD-Attention.

\subsection{Differential of the SVD}
Let $H_r = U \Sigma V^\top$ denote the rank-$r$ approximation of $H \in \mathbb{R}^{L \times d}$. Its differential is
\begin{equation}
\label{eq:diff_H_raw}
    \mathrm{d}H_{r} = \mathrm{d}U \Sigma V^\top + U \mathrm{d}\Sigma V^\top + U \Sigma \mathrm{d}V^\top.
\end{equation}
Because $U$ and $V$ are orthogonal, differentiating $U^\top U = I$ yields $(\mathrm{d}U)^\top U + U^\top \mathrm{d}U = 0$. Define $\Omega_U = U^\top \mathrm{d}U$ and $\Omega_V = V^\top \mathrm{d}V$; both matrices are skew-symmetric, and the differentials become
\begin{equation}
\label{eq:dU_dV_sub}
    \mathrm{d}U = U \Omega_U, \quad \mathrm{d}V = V \Omega_V.
\end{equation}
Substitution of Eq.~\eqref{eq:dU_dV_sub} into Eq.~\eqref{eq:diff_H_raw} gives
\begin{equation}
    \mathrm{d}H_{r} = U \Omega_U \Sigma V^\top + U \mathrm{d}\Sigma V^\top - U \Sigma \Omega_V V^\top.
\end{equation}
The projected differential $\mathrm{d}P$ isolates $\mathrm{d}\Sigma$, $\Omega_U$, and $\Omega_V$:
\begin{equation}
\label{eq:dP_def_full}
\begin{aligned}
    \mathrm{d}P \coloneqq U^\top \mathrm{d}H_{r} V = \Omega_U \Sigma + \mathrm{d}\Sigma - \Sigma \Omega_V.
\end{aligned}
\end{equation}

\subsection{Solving the Linear System}
The $(i,j)$-th entry of Eq.~\eqref{eq:dP_def_full} is
\begin{equation}
\label{eq:dP_elementwise}
    \mathrm{d}P_{ij} = (\Omega_U)_{ij} \sigma_j + \mathrm{d}\Sigma_{ij} - \sigma_i (\Omega_V)_{ij}.
\end{equation}
The diagonal entries of a skew-symmetric matrix are zero, and $\mathrm{d}\Sigma$ is diagonal; hence
\begin{equation}\label{eq:dP_cases}
    \mathrm{d}P_{ij} = 
    \begin{cases}
        \mathrm{d}\Sigma_{ii}, & i=j \\
        (\Omega_U)_{ij} \sigma_j - \sigma_i (\Omega_V)_{ij}, & i\neq j
    \end{cases}.
\end{equation}
For $i \neq j$, transposing the projection and using $\Omega^\top = -\Omega$ gives
\begin{equation}
\label{eq:sys_2}
    \mathrm{d}P_{ji} = (\mathrm{d}P^\top)_{ij} = -\sigma_i (\Omega_U)_{ij} + \sigma_j (\Omega_V)_{ij}.
\end{equation}
Equations~\eqref{eq:dP_cases} and~\eqref{eq:sys_2} form a linear system. The forward pass does not explicitly instantiate $U$, so eliminating $(\Omega_U)_{ij}$ yields
\begin{equation}
\label{eq:OmegaV_sol}
    (\Omega_V)_{ij} = \frac{\sigma_i \mathrm{d}P_{ij} + \sigma_j \mathrm{d}P_{ji}}{\sigma_j^2 - \sigma_i^2} = \frac{\sigma_i \mathrm{d}P_{ij} + \sigma_j \mathrm{d}P_{ji}}{-(\sigma_i^2 - \sigma_j^2)}.
\end{equation}

\subsection{Derivation of the Gradient Matrix}
We relate these differentials to the loss $\mathcal{L}$.
Let the upstream gradients be $\bar{U} = \partial\mathcal{L} / \partial U$, $\bar{\Sigma} = \partial\mathcal{L} / \partial \Sigma$, and $\bar{V} = \partial\mathcal{L} / \partial V$.
The total differential is
\begin{equation}
    \mathrm{d}\mathcal{L} = \operatorname{Tr}(\bar{U}^\top \mathrm{d}U) + \operatorname{Tr}(\bar{\Sigma}^\top \mathrm{d}\Sigma) + \operatorname{Tr}(\bar{V}^\top \mathrm{d}V).
\end{equation}
Because $U$ is absent from the forward pass, $\bar{U} \equiv 0$. Substituting $\mathrm{d}V = V \Omega_V$ gives
\begin{equation}
    \mathrm{d}\mathcal{L} = \operatorname{Tr}(\bar{\Sigma} \mathrm{d}\Sigma) + \operatorname{Tr}(P^\top \Omega_V),
\end{equation}
where $P = V^\top \bar{V}$. Expanding the trace sum:
\begin{equation}
    \mathrm{d}\mathcal{L} = \sum_{i} \bar{\Sigma}_{ii} \mathrm{d}\Sigma_{ii} + \sum_{i \neq j} P_{ji} (\Omega_V)_{ij}.
\end{equation}
Substituting Eq.~\eqref{eq:OmegaV_sol} yields
\begin{equation}
    \mathrm{d}\mathcal{L} = \sum_{i} \bar{\Sigma}_{ii} \mathrm{d}P_{ii} - \sum_{i \neq j} P_{ji} \left( \frac{\sigma_i \mathrm{d}P_{ij} + \sigma_j \mathrm{d}P_{ji}}{\sigma_i^2 - \sigma_j^2} \right).
\end{equation}
Swapping $i$ and $j$ in the second part of the summation exposes the coefficient of $\mathrm{d}P_{ij}$:
\begin{equation}
\begin{aligned}
    \text{Sum}_{i \neq j} &= \sum_{i \neq j} \left( - \frac{P_{ji} \sigma_i}{\sigma_i^2 - \sigma_j^2} \mathrm{d}P_{ij} - \frac{P_{ji} \sigma_j}{\sigma_i^2 - \sigma_j^2} \mathrm{d}P_{ji} \right) \\
    &= \sum_{i \neq j} \left( - \frac{P_{ji} \sigma_i}{\sigma_i^2 - \sigma_j^2} + \frac{P_{ij} \sigma_i}{\sigma_i^2 - \sigma_j^2} \right) \mathrm{d}P_{ij} \\
    &= \sum_{i \neq j} \frac{\sigma_i (P_{ij} - P_{ji})}{\sigma_i^2 - \sigma_j^2} \mathrm{d}P_{ij}.
\end{aligned}
\end{equation}
Define $G = \frac{\partial\mathcal{L}}{\partial P}$ so that $\mathrm{d}\mathcal{L} = \operatorname{Tr}(G^\top \mathrm{d}P)$. Its entries are
\begin{equation}
    G_{ij} = 
    \begin{cases}
        \bar{\Sigma}_{ii}, & i = j \\
        \sigma_i F_{ij} (P_{ij} - P_{ji}), & i \neq j
    \end{cases}
\end{equation}
where $F_{ij} = 1/(\sigma_i^2 - \sigma_j^2)$.

The implementation uses symmetric matrix operations. The derived term has the form $2 \sigma_i \operatorname{sym}(F \circ P)_{ij}$ because $F_{ji} = -F_{ij}$:
\begin{equation}
\begin{aligned}
    2 \operatorname{sym}(F \circ P)_{ij} = (F \circ P)_{ij} + (F \circ P)_{ji} = F_{ij} P_{ij} + F_{ji} P_{ji}= F_{ij} (P_{ij} - P_{ji}).
\end{aligned}
\end{equation}
Multiplication by $\sigma_i$ recovers the gradient form. With $G = \frac{\partial\mathcal{L}}{\partial P}$ and $\mathrm{d}P = U^\top \mathrm{d}H_{r} V$, the chain rule gives
\begin{equation}
\begin{aligned}
    \mathrm{d}\mathcal{L} = \operatorname{Tr}(G^\top \mathrm{d}P) = \operatorname{Tr}(G^\top U^\top \mathrm{d}H_{r} V)= \operatorname{Tr}((U G V^\top)^\top \mathrm{d}H_{r}).
\end{aligned}
\end{equation}
The gradient with respect to $H_r$ is therefore
\begin{equation}
\begin{aligned}
    \frac{\partial\mathcal{L}}{\partial H_r} = U G V^\top = U \left[ \bar{\Sigma}_{\operatorname{diag}} + 2 \Sigma \operatorname{sym}(F \circ (V^\top \bar{V})) \right] V^\top.
\end{aligned}
\end{equation}

\subsection{Bias Derivation}
\label{app:bias_derivation}
We analyze the bias of gradients derived from the rank-$r$ subspace relative to full-SVD gradients \citep{ionescu2015matrix}. Let the full SVD of $H \in \mathbb{R}^{L \times d}$ be $H = \mathbf{U} \mathbf{\Sigma} \mathbf{V}^\top$, where
\begin{equation}
    \mathbf{U} = [U \;|\; U_\perp],  
    \mathbf{\Sigma} = \begin{bmatrix} \Sigma & 0 \\ 0 & \Sigma_\perp \end{bmatrix}, 
    \mathbf{V} = [V \;|\; V_\perp].
\end{equation}
The complementary factors are $U_\perp\in\mathbb{R}^{L\times (L-r)}$, $\Sigma_\perp \in\mathbb{R}^{(L-r) \times(d-r)}$, and $V_\perp\in\mathbb{R}^{d\times(d-r)}$. Because $\bar{\mathbf{U}}\equiv 0$, the gradient of $\mathcal{L}$ with respect to $H$ is
\begin{equation}
\begin{aligned}
    \frac{\partial\mathcal{L}}{\partial H} = \mathbf{UGV^\top}=\mathbf{U} \left[ \bar{\mathbf\Sigma}_{\operatorname{diag}} + 2 \mathbf\Sigma \operatorname{sym}(\mathbf{F} \circ (\mathbf V^\top \bar{\mathbf V})) \right] \mathbf V^\top.
\end{aligned}
\end{equation}
We partition $\mathbf{G}$ into retained and complementary blocks:
\begin{equation}
    \mathbf{G} = \begin{bmatrix} G_{r,r} & G_{r,\perp} \\ G_{\perp,r} & G_{\perp,\perp} \end{bmatrix}.
\end{equation}
The full gradient expands as
\begin{equation}
\begin{aligned}
    \nabla_H \mathcal{L} &= [U \mid U_\perp] \begin{bmatrix} G_{r,r} & G_{r,\perp} \\ G_{\perp,r} & G_{\perp,\perp} \end{bmatrix} \begin{bmatrix} V^\top \\ V_\perp^\top \end{bmatrix} = U G_{r,r} V^\top + \operatorname{Res}.
\end{aligned}
\end{equation}
where
\begin{equation}
        \operatorname{Res} = U_\perp G_{\perp,r}V^\top + UG_{r,\perp}V_\perp^\top+U_\perp G_{\perp,\perp}V_\perp^\top.
\end{equation}
The first term equals $\partial\mathcal{L}/\partial H_r$, and the remaining blocks form the approximation bias. The blocks of $\mathbf{V}^\top\bar{\mathbf{V}}$ are
\begin{equation}
    \mathbf{V}^\top\bar{\mathbf{V}} = \begin{bmatrix}
        V^\top\\ V_\perp^\top
    \end{bmatrix} \begin{bmatrix}
        \bar{V}\ |\ \bar{V}_\perp
    \end{bmatrix} = \begin{bmatrix}
        V^\top\bar{V} & V^\top\bar{V}_\perp\\ V_\perp^\top \bar{V} & V_\perp^\top\bar{V}_\perp
    \end{bmatrix}
\end{equation}
The three residual blocks in $\mathbf{G}$ are therefore
\begin{equation}
\begin{aligned}
    G_{r,\perp} &= \Sigma \cdot \mathbf{F}_{r,\perp} \circ (V^\top\bar{V}_\perp + \bar{V}^\top V_\perp) \\
    G_{\perp, r} &= \Sigma_\perp \cdot \mathbf{F}_{\perp, r} \circ (V^\top\bar{V}_\perp + \bar{V}^\top V_\perp) \\
    G_{\perp,\perp} &= \bar{\Sigma}_\perp + 2\Sigma_\perp \cdot \mathbf{F}_{\perp,\perp} \circ \operatorname{sym}(V_\perp^\top\bar{V}_\perp)
\end{aligned}
\end{equation}
Two observations identify the dominant bias term:
\begin{enumerate}
    \item The loss depends only on the top-$r$ components, so $\bar{V}_\perp = 0$ and $\bar{\Sigma}_\perp = 0$, which gives $G_{\perp,\perp}=0$ and $V^\top\bar{V}_\perp = 0$.
    \item Assume the top-$r$ singular values dominate the remainder, i.e., $\sigma(\Sigma)\gg\sigma(\Sigma_\perp)$. Scaling by $\Sigma_\perp$ then makes $G_{\perp,r}$ negligible.
\end{enumerate}

For $G_{r,\perp}$ and indices $i \leq r$, $j > r$, the approximation $\sigma_i \gg \sigma_j \simeq 0$ yields
\begin{equation}
    (\mathbf{F}_{r,\perp})_{ij} = \frac{1}{\sigma_i^2 - \sigma_j^2} \simeq \frac{1}{\sigma_i^2} \implies \mathbf{F}_{r,\perp} \simeq \Sigma^{-2}.
\end{equation}
Substitution into $G_{r,\perp}$ gives
\begin{equation}
\begin{aligned}
    G_{r,\perp} \simeq \Sigma \cdot (\Sigma^{-2}) \circ (\bar{V}^\top V_\perp) = \Sigma^{-1} \bar{V}^\top V_\perp.
\end{aligned}
\end{equation}

Projecting the residual coupling term back to the input space gives the bias
\begin{equation}
    \mathcal{E} \simeq U G_{r,\perp} V_\perp^\top = U \Sigma^{-1} \bar{V}^\top (I_d - V V^\top).
\end{equation}
Let $\bar{V}_{\text{orth}} = \bar{V}^\top (I_d - V V^\top)$ denote the upstream-gradient component in the complementary subspace. Its contribution is bounded by
\begin{equation}
\begin{aligned}
    \|\mathcal{E}\|_F = \| U \Sigma^{-1} \bar{V}_{\text{orth}} \|_F \le \|U\|_2 \cdot \|\Sigma^{-1}\|_2 \cdot \|\bar{V}_{\text{orth}}\|_F = \frac{1}{\sigma_r} \|\bar{V}_{\text{orth}}\|_F,
\end{aligned}
\end{equation}
where $\sigma_r$ is the smallest retained singular value. As $\sigma_r \to 0$, $\Sigma^{-1}$ can amplify orthogonal noise and destabilize the full gradient. Truncated differentiation sets this residual contribution to zero and acts as a spectral regularizer outside the retained subspace.

\section{Softmax Approximation of SVD-Attention}
\label{app:svd_softmax_analysis}
{

Section~\ref{subsec:svd_attention} defines SVD-Attention by applying softmax to compact candidate-to-interest scores. Its underlying bilinear reduction uses $U^\top U=I_r$ exactly on a rank-$r$ history reconstruction, but nonlinear softmax prevents the same cancellation in token space. This appendix characterizes the residual between the normalized SVD-Attention output and token-level softmax attention. Table~\ref{tab:approximation_fidelity_full} complements Figure~\ref{fig:approximation_fidelity} by separating history truncation, softmax, and randomized-solver errors.

\begin{table}[H]
\centering

\caption{Detailed error decomposition for the approximation-fidelity study. Exact $H_r$ isolates rank truncation, SVD-Attention softmax isolates the domain change analyzed in Proposition~\ref{prop:svd_softmax_error}, and randomized SVD isolates solver error.}
\label{tab:approximation_fidelity_full}
\resizebox{\textwidth}{!}{
\begin{tabular}{llcccccc}
\toprule
\textbf{Method} & \textbf{Softmax domain} & $E_H\downarrow$ & $E_{\mathrm{total}}\downarrow$ & $E_{\mathrm{softmax}}\downarrow$ & $E_{\mathrm{solver}}\downarrow$ & \textbf{AUC}$\uparrow$ & \textbf{Latency (ms)}$\downarrow$ \\
\midrule
Full token softmax on $H$ & $N_L$ tokens & 0 & 0 & -- & -- & \textbf{0.8539} & $\sim1.3{\times}10^4$ \\
Token softmax on exact $H_r$ & $N_L$ tokens & \textbf{0.010} & \textbf{0.008} & 0 & -- & \textbf{0.8539} & $\sim1.35{\times}10^4$ \\
SVD-Attention (exact SVD) & $r$ components & \textbf{0.010} & \textbf{0.033} & \textbf{0.031} & 0 & \textbf{0.8538} & $\sim500$ \\
\textbf{SVD-Attention (randomized SVD)} & $r$ components & \textbf{0.013} & \textbf{0.035} & \textbf{0.031} & \textbf{0.002--0.004} & \textbf{0.8537} & \textbf{$\sim340$} \\
SVD bilinear computation (no softmax) & None & 0.013 & \textbf{0.14--0.18} & -- & -- & \textbf{0.8525--0.8527} & $\sim330$ \\
\bottomrule
\end{tabular}
}
\end{table}

Recall the rank-$r$ factorization
\begin{equation}
\mathrm{Key}=HW_K=U\Sigma V^\top W_K=U\mathrm{Key}_r ,\ 
\mathrm{Value}=HW_V=U\Sigma V^\top W_V=U\mathrm{Value}_r .
\end{equation}
The compact candidate-to-interest scores are
\begin{equation}
S=
\frac{\mathrm{Query}\,\mathrm{Key}_r^\top}{\sqrt d}
\in \mathbb{R}^{N_C\times r}.
\end{equation}
Then the rank-$r$ softmax attention that explicitly keeps $U$ can be written as
\begin{equation}
\begin{aligned}
\mathrm{Attn}_{U}
&=
\operatorname{softmax}
\left(
\frac{\mathrm{Query}\,\mathrm{Key}^\top}{\sqrt d}
\right)
\mathrm{Value}
\\
&=\operatorname{softmax}
\left(
\frac{\mathrm{Query}\,(U\mathrm{Key}_r)^\top}{\sqrt d}
\right)
U\mathrm{Value}_r
\\
&=\operatorname{softmax}
\left(
\frac{\mathrm{Query}\,\mathrm{Key}_r^\top U^\top}{\sqrt d}
\right)
U\mathrm{Value}_r
\\
&=\operatorname{softmax}(SU^\top)U\mathrm{Value}_r .
\label{eq:attn_with_u}
\end{aligned}
\end{equation}
SVD-Attention applies softmax directly to these $r$ scores:
\begin{equation}
\mathrm{Attn}_{\mathrm{SVD}}=\operatorname{softmax}(S)\mathrm{Value}_r .
\label{eq:attn_svd_softmax}
\end{equation}
The exact difference between keeping and eliminating $U$ is
\begin{equation}
\mathrm{Attn}_{U}
-\mathrm{Attn}_{\mathrm{SVD}}
=
\left[
\operatorname{softmax}(SU^\top)U-\operatorname{softmax}(S)
\right]
\mathrm{Value}_r .
\label{eq:u_residual_exact}
\end{equation}
A small residual in Eq.~\eqref{eq:u_residual_exact} requires
more than low rank. The following bound assumes that the principal components form
coherent principal-interest groups with limited within-group logit variation,
value variation, and multiplicity imbalance. These conditions yield the
residual bound below.

\paragraph{Principal-interest grouping.}
We model the dominant principal interests by partitioning historical behaviors
into $r$ principal-interest groups
\begin{equation}
\mathcal G_1,\ldots,\mathcal G_r,\qquad
n_a=|\mathcal G_a| .
\end{equation}
Define the within-interest logit variation as
\begin{equation}
\epsilon_K=\max_{c,a,t\in\mathcal G_a}
\left|(SU^\top)_{ct}-S_{ca}\right| ,
\label{eq:epsilon_k}
\end{equation}
where $c$ indexes a candidate row and $t$ indexes a historical behavior. This quantity measures how much the token-level logits expanded by $U^\top$ deviate from their corresponding principal-interest scores. Similarly, we define the within-interest value variation as
\begin{equation}
\epsilon_V=\max_{a,t\in\mathcal G_a}\left\|(U\mathrm{Value}_r)_{t:}-(\mathrm{Value}_r)_{a:}\right\|_2 .
\label{eq:epsilon_v}
\end{equation}
This quantity measures how much the token-level values expanded by $U$ deviate from their principal-interest values. Finally, we define the duplication-mass imbalance as
\begin{equation}
\epsilon_n=\min_{\alpha\in\mathbb R}\max_a
\left|\log n_a-\alpha\right| .
\label{eq:epsilon_n}
\end{equation}
If all principal interests have the same effective number of duplicated tokens, then $\epsilon_n=0$. If the group sizes are imbalanced,
$\epsilon_n$ measures the extra multiplicity prior introduced by token-level softmax.

\begin{proposition}[SVD-Attention Softmax Approximation]
\label{prop:svd_softmax_error}
Let $\|\mathrm{Value}_r\|_{\infty,2}=\max_a\|(\mathrm{Value}_r)_{a:}\|_2$. Under the grouping assumptions above, every candidate row $c$ satisfies
\begin{equation}
\left\|(\mathrm{Attn}_{U})_{c:}
-(\mathrm{Attn}_{\mathrm{SVD}})_{c:}\right\|_2
\le
\|\mathrm{Value}_r\|_{\infty,2}
\left(e^{2(\epsilon_K+\epsilon_n)}-1\right)+\epsilon_V.
\label{eq:softmax_svd_bound}
\end{equation}
\end{proposition}

\paragraph{Proof of Proposition~\ref{prop:svd_softmax_error}.}
Fix a candidate row $c$ and denote $s=S_{c:}\in\mathbb R^r$. The token-level and SVD-Attention softmax weights are
\begin{equation}
p=\operatorname{softmax}(sU^\top)\in\mathbb R^{N_L},\quad \hat p=\operatorname{softmax}(s)\in\mathbb R^r .
\end{equation}
The outputs are
\begin{equation}
\begin{aligned}
o_U
=
pU\mathrm{Value}_r
=
\sum_{t=1}^{N_L}
p_t
(U\mathrm{Value}_r)_{t:},\quad
o_{\mathrm{SVD}}
=
\hat p\,\mathrm{Value}_r
=
\sum_{a=1}^{r}
\hat p_a
(\mathrm{Value}_r)_{a:}.
\end{aligned}\end{equation}

For each group $\mathcal G_a$, let $P_a=\sum_{t\in\mathcal G_a}p_t$ denote its
total token-level softmax mass. Then
\begin{align}
o_U
&=\sum_{a=1}^{r}\sum_{t\in\mathcal G_a}p_t(U\mathrm{Value}_r)_{t:}\nonumber\\
&=\sum_{a=1}^{r}\sum_{t\in\mathcal G_a}p_t(\mathrm{Value}_r)_{a:}+\sum_{a=1}^{r}\sum_{t\in\mathcal G_a}p_t\left[(U\mathrm{Value}_r)_{t:}-(\mathrm{Value}_r)_{a:}\right]\nonumber\\
&=\sum_{a=1}^{r}P_a(\mathrm{Value}_r)_{a:}+R_V ,
\label{eq:ou_decompose}
\end{align}
where
\begin{equation}
R_V
=
\sum_{a=1}^{r}
\sum_{t\in\mathcal G_a}
p_t
\left[
(U\mathrm{Value}_r)_{t:}
-
(\mathrm{Value}_r)_{a:}
\right] .
\end{equation}
By the definition of $\epsilon_V$,
\begin{equation}
\|R_V\|_2
\le
\sum_{a=1}^{r}
\sum_{t\in\mathcal G_a}
p_t\epsilon_V
=
\epsilon_V .
\label{eq:rv_bound}
\end{equation}
Combining Eqs.~\eqref{eq:ou_decompose} and~\eqref{eq:rv_bound},
\begin{align}
o_U-o_{\mathrm{SVD}}
&=
\sum_{a=1}^{r}
(P_a-\hat p_a)
(\mathrm{Value}_r)_{a:}
+
R_V ,
\end{align}
which gives
\begin{equation}
\|o_U-o_{\mathrm{SVD}}\|_2
\le
\|\mathrm{Value}_r\|_{\infty,2}
\|P-\hat p\|_1
+
\epsilon_V .
\label{eq:reduce_to_mass}
\end{equation}
It remains to bound $\|P-\hat p\|_1$. For any token
$t\in\mathcal G_a$, Eq.~\eqref{eq:epsilon_k} implies
\begin{equation}
(sU^\top)_t=s_a+\eta_t,
\qquad
|\eta_t|\le\epsilon_K .
\end{equation}
Thus the total softmax mass of group $\mathcal G_a$ is
\begin{equation}
\begin{aligned}
P_a&=\frac{
\sum_{t\in\mathcal G_a}
\exp(s_a+\eta_t)
}{
\sum_{b=1}^{r}
\sum_{j\in\mathcal G_b}
\exp(s_b+\eta_j)
}\\
&=\frac{
\exp(s_a)n_a
\left(
\frac{1}{n_a}\sum_{t\in\mathcal G_a}\exp(\eta_t)
\right)
}{
\sum_{b=1}^{r}
\exp(s_b)n_b
\left(
\frac{1}{n_b}\sum_{j\in\mathcal G_b}\exp(\eta_j)
\right)
}.
\end{aligned}
\end{equation}
Define $\rho_a=\log\left(\frac{1}{n_a}\sum_{t\in\mathcal G_a}\exp(\eta_t)\right)$. Since $|\eta_t|\le\epsilon_K$, we have $-\epsilon_K\le \rho_a\le\epsilon_K$. Therefore,
\begin{equation}
P=\operatorname{softmax}(s+\log n+\rho),\quad\text{where}\quad \log n=(\log n_1,\ldots,\log n_r).
\label{eq:p_group_softmax}
\end{equation}
The SVD-Attention weights satisfy $\hat p=\operatorname{softmax}(s)$. By the definition of $\epsilon_n$, there exists a scalar $\alpha$ such
that
\begin{equation}
\|\log n-\alpha\mathbf 1\|_\infty\le\epsilon_n .
\end{equation}
Since softmax is invariant to adding the same scalar to every logit,
\begin{equation}
\operatorname{softmax}(s+\log n+\rho)
=
\operatorname{softmax}(s+\log n-\alpha\mathbf 1+\rho).
\end{equation}
Let $h=\log n-\alpha\mathbf 1+\rho$; then
$\|h\|_\infty\le\epsilon_n+\epsilon_K$. A standard softmax perturbation bound states that, for
$q=\operatorname{softmax}(x+h)$ and $p=\operatorname{softmax}(x)$, if
$\|h\|_\infty\le\delta$, then
\begin{equation}
\|q-p\|_1\le e^{2\delta}-1 .
\label{eq:softmax_perturb}
\end{equation}
Indeed,
\begin{equation}
\frac{q_i}{p_i}
=
\frac{\exp(h_i)}
{\sum_j p_j\exp(h_j)} .
\end{equation}
Since $e^{-\delta}\le \exp(h_i)\le e^\delta$ and
$e^{-\delta}\le\sum_j p_j\exp(h_j)\le e^\delta$, we obtain
\begin{equation}
e^{-2\delta}
\le
\frac{q_i}{p_i}
\le
e^{2\delta}.
\end{equation}
Hence
\begin{equation}
|q_i-p_i|
\le
(e^{2\delta}-1)p_i .
\end{equation}
Summing over $i$ proves Eq.~\eqref{eq:softmax_perturb}. Taking $x=s$ and $\delta=\epsilon_K+\epsilon_n$ yields
\begin{equation}
\|P-\hat p\|_1
\le
e^{2(\epsilon_K+\epsilon_n)}-1 .
\label{eq:mass_bound}
\end{equation}
Substituting Eq.~\eqref{eq:mass_bound} into
Eq.~\eqref{eq:reduce_to_mass} gives
\begin{equation}
\|o_U-o_{\mathrm{SVD}}\|_2
\le
\|\mathrm{Value}_r\|_{\infty,2}
\left(
e^{2(\epsilon_K+\epsilon_n)}-1
\right)
+
\epsilon_V .
\end{equation}
Since the argument holds for every candidate row $c$, the row-wise bound follows. This completes the proof.

The residual in Eq.~\eqref{eq:softmax_svd_bound} has three sources. The term $\epsilon_K$ measures within-interest score variation after expanding principal interests to token-level logits through $U^\top$; it is small when behaviors associated with the same interest have similar keys. The term $\epsilon_V$ measures within-interest value variation in $U\mathrm{Value}_r$ and is small when repeated or semantically similar behaviors have close value representations. Finally, $\epsilon_n$ measures duplication-mass imbalance. Full token-level softmax assigns an interest group with more repeated tokens a multiplicity bonus of approximately $\log n_a$, whereas SVD-Attention applies softmax directly to principal interests and omits this term. Balanced effective group masses keep $\epsilon_n$ small and the approximation tight. Under imbalanced masses, the omitted term equals the repeated-token counting bias of full softmax attention. Such mass often reflects redundant exposures or short-term repeated behaviors in lifelong histories, so its omission can act as de-duplication and denoising.
}

%% file: appdix/Appdix-Theory-Detail.tex
\section{Formal Set-Wise Analysis}
\label{app:setwise_formal}

{
This appendix gives the formal motivation for SOLAR's set-wise component. The results concern candidate-set conditioning and do not alter the exact bilinear SVD identity in Section~\ref{subsec:svd_attention} or the SVD-Attention softmax analysis in Appendix~\ref{app:svd_softmax_analysis}.

\subsection{Context-Dependent Ranking}

\begin{definition}[Contextual Flip]
\label{assump:contextual_flip}
There exist items $x_i,x_j\in\mathcal X$ and two candidate sets $X^{(1)},X^{(2)}\subset\mathcal X$, each occurring with non-zero probability, such that
\begin{equation}
\eta^\star(x_i,X^{(1)})>\eta^\star(x_j,X^{(1)}),
\qquad
\eta^\star(x_i,X^{(2)})<\eta^\star(x_j,X^{(2)}).
\end{equation}
\end{definition}

For a point-wise scorer $f\in\mathcal F_{\mathrm{point}}$, let
$\pi^\star_{ij}(X)=\mathbb I(\eta^\star(x_i,X)>\eta^\star(x_j,X))$ and define the pairwise surrogate population risk
\begin{equation}
\label{eq:pair_risk_point}
R(f)=\mathbb E_X\!\left[\sum_{i\ne j}\mathcal L_{\mathrm{BCE}}\big(f(x_i)-f(x_j),\pi^\star_{ij}(X)\big)\right].
\end{equation}

\begin{theorem}[Bayes Limit of Point-Wise Scorers]
\label{thm:pairwise_point_bayes}
Let
$p_{ij}=\mathbb P(\eta^\star(x_i,X)>\eta^\star(x_j,X))$.
For $p_{ij}\in(0,1)$, the minimizer of Eq.~\eqref{eq:pair_risk_point} satisfies
\begin{equation}
\label{eq:pairwise_bayes_solution}
f^\star(x_i)-f^\star(x_j)=\log\frac{p_{ij}}{1-p_{ij}}.
\end{equation}
A point-wise scorer therefore represents the marginal preference probability, which cannot vary with the candidate set.
\end{theorem}

\begin{corollary}[Irreducible Risk under Contextual Flip]
\label{cor:context_flip_pairwise}
Under Definition~\ref{assump:contextual_flip}, every context-independent point-wise scorer incurs strictly positive Bipartite Ranking Risk on at least one of the two contexts. A realizable set-wise scorer of the form $f_{\mathrm{set}}(x,X)=\eta^\star(x,X)$ can represent both orderings.
\end{corollary}

\subsection{Correlation-Dependent Generalization Bounds}

Consider $N$ independent requests, each containing $m$ candidates. Let
$s_{u,i}=w^\top\phi(u,x_i,X_u)$ with $\|w\|_2\le W$ and
$\|\phi(\cdot)\|_2\le B$. For a $1$-Lipschitz loss taking values in $[0,1]$, define
\begin{equation}
R(f)=\mathbb E_{\mathcal D}[\ell(f(u,x_i,X_u),y_{u,i})],
\qquad
\hat R_S(f)=\frac{1}{|S|}\sum_{(u,i)\in S}\ell(f(u,x_{u,i},X_u),y_{u,i}).
\end{equation}
The standard Rademacher bound gives, with probability at least $1-\delta$,
\begin{equation}
\label{eq:std_gen}
R(f)\le\hat R_S(f)+2\mathfrak R_S(\ell\circ\mathcal F)
+3\sqrt{\frac{\ln(2/\delta)}{2T}},
\end{equation}
where $T$ is the number of independent sampling units.

\begin{theorem}[I.I.D. Candidate Samples]
\label{thm:point}
If all $(u,i)$ samples are i.i.d., then
\begin{equation}
\label{eq:rad_point}
\mathfrak R_S(\ell\circ\mathcal F)\le\frac{WB}{\sqrt{mN}},
\end{equation}
and, with probability at least $1-\delta$,
\begin{equation}
\label{eq:gen_point_thm}
R(f)-\hat R_S(f)
\le\frac{2WB}{\sqrt{mN}}
+3\sqrt{\frac{\ln(2/\delta)}{2mN}}.
\end{equation}
\end{theorem}

\begin{theorem}[Block-Dependent Candidate Samples]
\label{thm:block}
Assume that requests are independent and that the representations within request $u$ satisfy
\begin{equation}
\label{eq:rho_assump}
\frac{1}{m(m-1)}\sum_{i\ne j}z_{u,i}^\top z_{u,j}
\le\rho B^2,
\qquad \rho\in[0,1],
\end{equation}
where $z_{u,i}=\phi(u,x_i,X_u)$. Then
\begin{equation}
\label{eq:rad_block}
\mathfrak R_S(\ell\circ\mathcal F)
\le\frac{WB}{\sqrt{mN}}\sqrt{1+(m-1)\rho},
\end{equation}
and, with probability at least $1-\delta$,
\begin{equation}
\label{eq:gen_block_thm}
R(f)-\hat R_S(f)
\le\frac{2WB}{\sqrt{mN}}\sqrt{1+(m-1)\rho}
+3\sqrt{\frac{\ln(2/\delta)}{2N}}.
\end{equation}
\end{theorem}

\begin{corollary}[Correlation Factor]
\label{cor:mismatch}
Under the preceding assumptions, the correlation-dependent multiplicative factor is
\begin{equation}
\mathcal C(m,\rho)=\sqrt{1+(m-1)\rho}.
\end{equation}
The bound becomes looser as the assumed within-request correlation increases; it does not establish that a particular learned architecture reduces $\rho$.
\end{corollary}

For the positive set $\mathcal P_u=\{i:y_{u,i}=1\}$, define the list-wise loss
\begin{equation}
\label{eq:listwise_loss_def}
\ell_{\mathrm{list}}(s;\mathcal P_u)
=-\frac{1}{|\mathcal P_u|}\sum_{i\in\mathcal P_u}s_i
+\log\!\left(\sum_{j=1}^m\exp(s_j)\right).
\end{equation}

\begin{lemma}[Lipschitz Continuity]
\label{lem:listwise_lipschitz}
$\ell_{\mathrm{list}}$ is $L$-Lipschitz with respect to $s$ in the $\ell_2$ norm, where $L\le\sqrt 2$.
\end{lemma}

\begin{corollary}[List-Wise Generalization Gap]
\label{thm:listwise_gap}
Under the block-dependent assumptions above,
\begin{equation}
\label{eq:listwise_final_bound}
R(f)-\hat R_S(f)
\lesssim\frac{2\sqrt 2WB}{\sqrt{mN}}\mathcal C(m,\rho)
+\mathcal O\!\left(\frac{1}{\sqrt N}\right).
\end{equation}
\end{corollary}

\subsection{Correlation under an Explicit Projection Model}

Let two candidates share a dominant component,
$x_k=\mathbf c+\mathbf d_k$ for $k\in\{i,j\}$, with
$\|\mathbf c\|\gg\|\mathbf d_k\|$.

\begin{theorem}[Correlation under the Projection Assumptions]
\label{thm:corr}
Assume that the point-wise mapping is Lipschitz with normalized outputs. Further assume that the set-wise interaction implements a projection that removes the shared component $\mathbf c$ and that the remaining discriminative components are asymptotically orthogonal. Under these explicit assumptions, the point-wise output correlation remains lower-bounded by input smoothness, whereas the projected set-wise representations can approach zero correlation:
\begin{equation}
\rho_{\mathrm{set}}\ll\rho_{\mathrm{point}}.
\end{equation}
This conditional result motivates candidate interaction. It does not establish that every learned set-wise architecture realizes the assumed projection.
\end{theorem}
}

\section{Proof Details}
\subsection{Proof of Theorem~\ref{thm:pairwise_point_bayes}}
\label{pf:pairwise_point_bayes}
\begin{proof}
The population risk $R(f)$ can be decomposed into independent optimization problems for each pair $(i, j)$. Let $\Delta_{ij} = f(x_i) - f(x_j)$ be the logit difference. Minimizing Eq.~\eqref{eq:pair_risk_point} is equivalent to minimizing each pair:
{\small
\begin{equation}
\begin{aligned}
\min_{\Delta_{ij}}\ & \mathbb{E}_X [\pi^\star_{ij}(X)\log(1+e^{-\Delta_{ij}}) + (1-\pi^\star_{ij}(X))\log(1+e^{\Delta_{ij}})] \\
& = p_{ij}\log(1+e^{-\Delta_{ij}}) + (1-p_{ij})\log(1+e^{\Delta_{ij}}).
\end{aligned}
\end{equation}
}
This objective is strictly convex with respect to $\Delta_{ij}$. The optimality condition is:
\begin{equation}
\frac{\partial \text{Obj}}{\partial \Delta_{ij}} = -p_{ij} \frac{e^{-\Delta_{ij}}}{1+e^{-\Delta_{ij}}} + (1-p_{ij}) \frac{e^{\Delta_{ij}}}{1+e^{\Delta_{ij}}} = 0.
\end{equation}
With $\sigma(z) = \frac{1}{1+e^{-z}}$, this simplifies to:
\begin{equation}
\begin{aligned}
-p_{ij}(1-\sigma(\Delta_{ij})) + (1-p_{ij})\sigma(\Delta_{ij}) = 0
\implies \sigma(\Delta_{ij}) = p_{ij}.
\end{aligned}
\end{equation}
Solving for $\Delta_{ij}$ yields $\log (p_{ij}/(1-p_{ij}))$.
\end{proof}

\subsection{Proof of Corollary~\ref{cor:context_flip_pairwise}}
\label{pf:context_flip_pairwise}
\begin{proof}
Consider any point-wise scorer $f \in \mathcal{F}_{\mathrm{point}}$. By definition, the score $f(x)$ depends solely on item $x$, implying that the predicted score difference $\Delta_{ij} = f(x_i) - f(x_j)$ is a constant scalar independent of the context $X$. Consequently, the predicted ordering $\mathrm{sgn}(\Delta_{ij})$ is static.

Under Definition~\ref{assump:contextual_flip}, the ground-truth ordering $\mathrm{sgn}(\eta^\star(x_i, X) - \eta^\star(x_j, X))$ takes opposite values in contexts $X^{(1)}$ and $X^{(2)}$. Any fixed value of $f$ therefore misclassifies the preference relation in at least one context.

Since the misclassification event occurs with probability at least $\min(\mathbb{P}(X^{(1)}), \mathbb{P}(X^{(2)})) > 0$, the indicator function in the definition of Bipartite Ranking Risk (Definition~\ref{def:ranking_error}) contributes a non-zero value to the expectation. Thus, $\mathcal{R}(f) > 0$.
\end{proof}

\subsection{Proof of Theorem~\ref{thm:point}}
\label{pf:point}
\begin{proof}
We first bound the empirical Rademacher complexity $\hat{\mathfrak{R}}_S(\ell\circ\mathcal{F}_{\mathrm{point}})$. 
Because $\ell$ is $1$-Lipschitz, Talagrand's contraction lemma gives
\begin{equation}
\hat{\mathfrak{R}}_S(\ell\circ\mathcal{F}_{\mathrm{point}})
\;\le\;
\hat{\mathfrak{R}}_S(\mathcal{F}_{\mathrm{point}}).
\end{equation}
Let $\{\sigma_{u,i}\}$ be i.i.d. Rademacher variables. Following the definition of $\mathcal{F}_{\mathrm{point}}$, let $z_{u,i} = \phi(u, x_{u,i})$ denote the latent representation. The empirical Rademacher complexity of $\mathcal{F}_{\mathrm{point}}$ is:
\begin{align}
\hat{\mathfrak{R}}_S(\mathcal{F}_{\mathrm{point}})
&=
\mathbb{E}_{\sigma}\Big[
\sup_{\|w\|_2\le W}
\frac{1}{mN}\sum_{u=1}^N\sum_{i=1}^m \sigma_{u,i}\, w^\top z_{u,i}
\Big] \nonumber\\
&=
\frac{1}{mN}\mathbb{E}_{\sigma}\Big[
\sup_{\|w\|_2\le W}
w^\top \Big(\sum_{u,i}\sigma_{u,i} z_{u,i}\Big)
\Big] \nonumber\\
&=
\frac{W}{mN}\mathbb{E}_{\sigma}\Big[
\Big\|\sum_{u,i}\sigma_{u,i} z_{u,i}\Big\|_2
\Big],
\end{align}
where the last equality follows from 
\begin{equation}
\sup_{\|w\|_2\le W} w^\top v = W\|v\|_2.
\end{equation}
Applying Jensen's inequality, we obtain:
\begin{align}
\hat{\mathfrak{R}}_S(\mathcal{F}_{\mathrm{point}})
&\le
\frac{W}{mN}
\sqrt{
\mathbb{E}_{\sigma}\Big[
\Big\|\sum_{u,i}\sigma_{u,i} z_{u,i}\Big\|_2^2
\Big]
} \nonumber\\
&=
\frac{W}{mN}
\sqrt{
\mathbb{E}_{\sigma}\Big[
\sum_{u,i}\sum_{v,j}\sigma_{u,i}\sigma_{v,j}\, z_{u,i}^\top z_{v,j}
\Big]
}.
\end{align}
Since $\sigma_{u,i}$ are independent with zero mean, the expectation $\mathbb{E}[\sigma_{u,i}\sigma_{v,j}]$ is $1$ if $(u,i)=(v,j)$ and $0$ otherwise. Thus:
\begin{equation}
\mathbb{E}_{\sigma}\Big[
\Big\|\sum_{u,i}\sigma_{u,i} z_{u,i}\Big\|_2^2
\Big]
=
\sum_{u=1}^N\sum_{i=1}^m \|z_{u,i}\|_2^2
\le mNB^2,
\end{equation}
Substitution gives $\hat{\mathfrak{R}}_S(\mathcal{F}_{\mathrm{point}}) \le \frac{W}{mN}\sqrt{mNB^2} = \frac{WB}{\sqrt{mN}}$. Taking the expectation over $S$ yields Eqs.~\eqref{eq:rad_point} and~\eqref{eq:gen_point_thm}.
\end{proof}

\subsection{Proof of Theorem~\ref{thm:block}}
\label{pf:block}
\begin{proof}
The same contraction argument used in Theorem~\ref{thm:point} gives
\begin{equation}
\hat{\mathfrak{R}}_S(\ell\circ\mathcal{F}_{\mathrm{set}})
\;\le\;
\hat{\mathfrak{R}}_S(\mathcal{F}_{\mathrm{set}}).
\end{equation}
Here, requests are independent and the items within each candidate set $X_u$ are dependent. We therefore use one Rademacher variable per request, $\{\sigma_u\}_{u=1}^N$, and define the set-wise empirical Rademacher complexity as
\begin{align}
\hat{\mathfrak{R}}_S(\mathcal{F}_{\mathrm{set}})
&=
\mathbb{E}_{\sigma}\Big[
\sup_{\|w\|_2\le W}
\frac{1}{mN}\sum_{u=1}^N \sigma_u \sum_{i=1}^m w^\top z_{u,i}
\Big] \nonumber\\
&=
\frac{1}{mN}\mathbb{E}_{\sigma}\Big[
\sup_{\|w\|_2\le W}
w^\top \Big(\sum_{u=1}^N \sigma_u \underbrace{\sum_{i=1}^m z_{u,i}}_{=: \bar{Z}_u}\Big)
\Big] \nonumber\\
&=
\frac{W}{mN}\mathbb{E}_{\sigma}\Big[
\Big\|\sum_{u=1}^N \sigma_u \bar{Z}_u\Big\|_2
\Big].
\end{align}
Applying Jensen's inequality and expanding the squared norm:
\begin{align}
\hat{\mathfrak{R}}_S(\mathcal{F}_{\mathrm{set}})
\le
\frac{W}{mN}
\sqrt{
\mathbb{E}_{\sigma}\Big[
\Big\|\sum_{u=1}^N \sigma_u \bar{Z}_u\Big\|_2^2
\Big]
} \nonumber\\
&=
\frac{W}{mN}
\sqrt{
\mathbb{E}_{\sigma}\Big[
\sum_{u=1}^N\sum_{v=1}^N \sigma_u\sigma_v\, \bar{Z}_u^\top \bar{Z}_v
\Big]
}.
\end{align}
The variables $\sigma_u$ are independent across requests and have zero mean, so the cross terms vanish for $u\neq v$ and $\mathbb{E}[\sigma_u^2]=1$. Hence
\begin{equation}
\mathbb{E}_{\sigma}\Big[
\Big\|\sum_{u=1}^N \sigma_u \bar{Z}_u\Big\|_2^2
\Big]
=
\sum_{u=1}^N \|\bar{Z}_u\|_2^2
=
\sum_{u=1}^N \Big\|\sum_{i=1}^m z_{u,i}\Big\|_2^2.
\end{equation}

The latent-correlation assumption bounds $\|\sum_{i=1}^m z_{u,i}\|_2^2$ for each request $u$:
\begin{align}
\Big\|\sum_{i=1}^m z_{u,i}\Big\|_2^2
=
\sum_{i=1}^m \|z_{u,i}\|_2^2 + \sum_{i\neq j} z_{u,i}^\top z_{u,j} \nonumber\\
&\le mB^2 + m(m-1)\rho B^2 \nonumber\\
&=
mB^2\big(1+(m-1)\rho\big),
\end{align}
where we used the bound $\|z_{u,i}\|_2 \le B$ and the correlation condition \eqref{eq:rho_assump}. Substituting this back into the complexity bound:
\begin{align}
\hat{\mathfrak{R}}_S(\mathcal{F}_{\mathrm{set}})
\le
\frac{W}{mN}\sqrt{
\sum_{u=1}^N mB^2\big(1+(m-1)\rho\big)
} \nonumber\\
&=\frac{W}{mN}\sqrt{NmB^2\big(1+(m-1)\rho\big)} \nonumber\\
&=\frac{WB}{\sqrt{mN}}\sqrt{1+(m-1)\rho}.
\end{align}
Taking the expectation over $S$ yields Eq.~\eqref{eq:rad_block}. Substituting this bound and $T=N$, the number of independent requests, into Eq.~\eqref{eq:std_gen} gives Eq.~\eqref{eq:gen_block_thm}.
\end{proof}

\subsection{Proof of Lemma~\ref{lem:listwise_lipschitz}}
\label{pf:lipschitz}
\begin{proof}
Let $k = |\mathcal{P}_u|$ and $\pi_i = \frac{\exp(s_i)}{\sum_{j=1}^m \exp(s_j)}$, the softmax probability of item $i$. The derivative of $\ell_{\mathrm{list}}$ with respect to $s_t$ is
\begin{equation}
\frac{\partial \ell_{\mathrm{list}}}{\partial s_t} = 
\begin{cases}
-\frac{1}{k} + \pi_t, & \text{if } t \in \mathcal{P}_u \text{ (positive item)}, \\
\pi_t, & \text{if } t \notin \mathcal{P}_u \text{ (negative item)}.
\end{cases}
\end{equation}
The squared $\ell_2$ norm of $\nabla_s \ell$ is
\begin{equation}
\begin{aligned}
\|\nabla_s \ell\|_2^2 
&= \sum_{t \in \mathcal{P}_u} \left( -\frac{1}{k} + \pi_t \right)^2 + \sum_{t \notin \mathcal{P}_u} (\pi_t)^2 = \sum_{t \in \mathcal{P}_u} \left( \frac{1}{k^2} - \frac{2}{k}\pi_t + \pi_t^2 \right) + \sum_{t \notin \mathcal{P}_u} \pi_t^2 \\
&= \sum_{t \in \mathcal{P}_u} \frac{1}{k^2} - \frac{2}{k} \sum_{t \in \mathcal{P}_u} \pi_t + \sum_{t=1}^m \pi_t^2 = k \cdot \frac{1}{k^2} - \frac{2}{k} \pi_{\mathrm{pos}} + \|\pi\|_2^2 = \frac{1}{k} - \frac{2}{k} \pi_{\mathrm{pos}} + \|\pi\|_2^2,
\end{aligned}
\end{equation}
where $\pi_{\mathrm{pos}} = \sum_{t \in \mathcal{P}_u} \pi_t \in [0,1]$ is the total probability mass on positive items.
Because $\|\pi\|_2^2 \le \sum_i \pi_i = 1$ and $\pi_{\mathrm{pos}} \ge 0$,
\begin{equation}
\|\nabla_s \ell\|_2^2 \le \frac{1}{k} + 1 \le 2 \quad (\text{since } k \ge 1).
\end{equation}
Thus, $\|\nabla_s \ell\|_2 \le \sqrt{2}$ for any number of positive items, proving the Lipschitz bound.
\end{proof}

\subsection{Proof of Theorem~\ref{thm:corr}}
\label{pf:corr}
\begin{proof}
\textbf{Point-wise model.}
Let ${z}_k = \phi_\theta(u,{x}_k)$, where $\phi_\theta$ is continuous. A deep network with bounded weights is Lipschitz continuous, so there exists $K$ such that
\begin{equation}
    \| {z}_i - {z}_j \| = \| \phi_\theta({x}_i) - \phi_\theta({x}_j) \| \le K \| {x}_i - {x}_j \|
\end{equation}
The input decomposition gives $\| {x}_i - {x}_j \| = \| {d}_i - {d}_j \|$. Small discriminative variations therefore constrain the distance between outputs. For normalized outputs, $\|{z}_i\|=\|z_j\| =1$, and $\| {z}_i - {z}_j \|^2 = 2(1 - \rho_{\text{point}})$. The Lipschitz condition implies
\begin{equation}
    2(1 - \rho_{\text{point}}) \le K^2 \| {d}_i - {d}_j \|^2 = \epsilon
\end{equation}
Thus, $\rho_{\text{point}} \ge 1 - \frac{\epsilon}{2} \to 1$, and the point-wise outputs remain highly correlated.

\textbf{Set-wise model.}
Let ${z}_k = \phi_\theta(u,{x}_i, X_u)$, where $X_u$ is the candidate set. Under the projection model in the theorem, set-wise attention computes
\begin{equation}
    {z}_k = W_V \left( {x}_k - \sum_{i \in X_u} \alpha_{k,i} {x}_i \right)
\end{equation}
The aggregated context approximates the shared component $\mathbf{c}$, yielding a projection onto its orthogonal complement:
\begin{equation}
    {z}_i \to \mathcal{P}_{\mathbf{c}^\perp}({x}_i) = \mathbf{d}_i, \quad {z}_j \to \mathcal{P}_{\mathbf{c}^\perp}({x}_j) = \mathbf{d}_j
\end{equation}
After removing the dominant shared component, the outputs are determined by $\mathbf{d}_i$ and $\mathbf{d}_j$. Under the assumed asymptotic orthogonality, $\mathbf{d}_i \perp \mathbf{d}_j$, and the resulting correlation is
\begin{equation}
    \rho_{\text{set}} \propto \mathbf{d}_i^\top \mathbf{d}_j \to 0
\end{equation}
Under the explicit projection and orthogonality assumptions above, this construction yields $\rho_{\mathrm{set}}\ll\rho_{\mathrm{point}}$. The conditional construction motivates set-wise interaction. It does not imply that an arbitrary learned set-wise architecture necessarily realizes the projection or reduces empirical representation correlation.
\end{proof}

\section{Construction of Datasets}
\label{datasets}
\paragraph{RecFlow.}
RecFlow \citep{liu2025recflow} logs user behavior throughout a multi-stage industrial recommendation pipeline, exposing candidate sets at intermediate stages in addition to final impressions. We follow the dataset protocol and use the second period, from February 5 to February 18, which contains 3,308,233 requests. Evaluation occurs at the re-ranking stage. For each request, we take the top 120 items from the upstream ranking module and retain their original order as a positional reference feature. The model consumes item-side attributes, including video ID, category, and author-related fields such as upload type, age, gender, and device type, together with a length-50 user-context sequence of recent interactions.

\paragraph{MIND.}
MIND \citep{wu-etal-2020-mind} is a large-scale benchmark derived from Microsoft News impression logs. Each labeled sample follows the format $[uID, t, ClickHist, ImpLog]$, where $ImpLog$ lists the displayed news articles and click labels, and $ClickHist$ records previously clicked news IDs in chronological order. Each article includes a news ID, title, abstract, body, and editor-assigned category; the release also provides WikiData-linked entities, knowledge triples, and embeddings. We treat each impression as a request and use the displayed articles in $ImpLog$ as the base candidate set. We augment these exposures with negatives drawn from the full news pool: uniform samples, globally popular articles, and similar items.

%% file: references.bib
@article{ionescu2015matrix,
  author       = {Catalin Ionescu and
                  Orestis Vantzos and
                  Cristian Sminchisescu},
  title        = {Training Deep Networks with Structured Layers by Matrix Backpropagation},
  journal      = {CoRR},
  volume       = {abs/1509.07838},
  year         = {2015},
  url          = {http://arxiv.org/abs/1509.07838},
  eprinttype    = {arXiv},
  eprint       = {1509.07838},
  bibsource    = {dblp computer science bibliography, https://dblp.org}
}

@article{yu2024ifainteractionfidelityattention,
  author       = {Wenhui Yu and
                  Chao Feng and
                  Yanze Zhang and
                  Lantao Hu and
                  Peng Jiang and
                  Han Li},
  title        = {{IFA:} Interaction Fidelity Attention for Entire Lifelong Behaviour
                  Sequence Modeling},
  journal      = {CoRR},
  volume       = {abs/2406.09742},
  year         = {2024},
  url          = {https://doi.org/10.48550/arXiv.2406.09742},
  doi          = {10.48550/ARXIV.2406.09742},
  eprinttype    = {arXiv},
  eprint       = {2406.09742},
  bibsource    = {dblp computer science bibliography, https://dblp.org}
}

@article{menon2016bipartite,
  author       = {Aditya Krishna Menon and
                  Robert C. Williamson},
  title        = {Bipartite Ranking: a Risk-Theoretic Perspective},
  journal      = {J. Mach. Learn. Res.},
  volume       = {17},
  pages        = {195:1--195:102},
  year         = {2016},
  url          = {https://jmlr.org/papers/v17/14-265.html},
  bibsource    = {dblp computer science bibliography, https://dblp.org}
}

@inproceedings{zhou2018deepnetworkclickthroughrate,
  author       = {Guorui Zhou and
                  Xiaoqiang Zhu and
                  Chengru Song and
                  Ying Fan and
                  Han Zhu and
                  Xiao Ma and
                  Yanghui Yan and
                  Junqi Jin and
                  Han Li and
                  Kun Gai},
  editor       = {Yike Guo and
                  Faisal Farooq},
  title        = {Deep Interest Network for Click-Through Rate Prediction},
  booktitle    = {Proceedings of the 24th {ACM} {SIGKDD} International Conference on
                  Knowledge Discovery {\&} Data Mining, {KDD} 2018, London, UK,
                  August 19-23, 2018},
  pages        = {1059--1068},
  publisher    = {{ACM}},
  year         = {2018},
  url          = {https://doi.org/10.1145/3219819.3219823},
  doi          = {10.1145/3219819.3219823},
  bibsource    = {dblp computer science bibliography, https://dblp.org}
}

@inproceedings{ESMM,
  author       = {Xiao Ma and
                  Liqin Zhao and
                  Guan Huang and
                  Zhi Wang and
                  Zelin Hu and
                  Xiaoqiang Zhu and
                  Kun Gai},
  editor       = {Kevyn Collins{-}Thompson and
                  Qiaozhu Mei and
                  Brian D. Davison and
                  Yiqun Liu and
                  Emine Yilmaz},
  title        = {Entire Space Multi-Task Model: An Effective Approach for Estimating
                  Post-Click Conversion Rate},
  booktitle    = {The 41st International {ACM} {SIGIR} Conference on Research {\&}
                  Development in Information Retrieval, {SIGIR} 2018, Ann Arbor, MI,
                  USA, July 08-12, 2018},
  pages        = {1137--1140},
  publisher    = {{ACM}},
  year         = {2018},
  url          = {https://doi.org/10.1145/3209978.3210104},
  doi          = {10.1145/3209978.3210104},
  bibsource    = {dblp computer science bibliography, https://dblp.org}
}

@inproceedings{WangSCJLHC21,
  author       = {Ruoxi Wang and
                  Rakesh Shivanna and
                  Derek Zhiyuan Cheng and
                  Sagar Jain and
                  Dong Lin and
                  Lichan Hong and
                  Ed H. Chi},
  editor       = {Jure Leskovec and
                  Marko Grobelnik and
                  Marc Najork and
                  Jie Tang and
                  Leila Zia},
  title        = {{DCN} {V2:} Improved Deep {\&} Cross Network and Practical Lessons
                  for Web-scale Learning to Rank Systems},
  booktitle    = {{WWW} '21: The Web Conference 2021, Virtual Event / Ljubljana, Slovenia,
                  April 19-23, 2021},
  pages        = {1785--1797},
  publisher    = {{ACM} / {IW3C2}},
  year         = {2021},
  url          = {https://doi.org/10.1145/3442381.3450078},
  doi          = {10.1145/3442381.3450078},
  bibsource    = {dblp computer science bibliography, https://dblp.org}
}

@article{halko2011finding,
  author       = {Nathan Halko and
                  Per{-}Gunnar Martinsson and
                  Joel A. Tropp},
  title        = {Finding Structure with Randomness: Probabilistic Algorithms for Constructing
                  Approximate Matrix Decompositions},
  journal      = {{SIAM} Rev.},
  volume       = {53},
  number       = {2},
  pages        = {217--288},
  year         = {2011},
  url          = {https://doi.org/10.1137/090771806},
  doi          = {10.1137/090771806},
  bibsource    = {dblp computer science bibliography, https://dblp.org}
}

@misc{vaswani2023attentionneed,
  author       = {Ashish Vaswani and
                  Noam Shazeer and
                  Niki Parmar and
                  Jakob Uszkoreit and
                  Llion Jones and
                  Aidan N. Gomez and
                  Lukasz Kaiser and
                  Illia Polosukhin},
  editor       = {Isabelle Guyon and
                  Ulrike von Luxburg and
                  Samy Bengio and
                  Hanna M. Wallach and
                  Rob Fergus and
                  S. V. N. Vishwanathan and
                  Roman Garnett},
  title        = {Attention is All you Need},
  booktitle    = {Advances in Neural Information Processing Systems 30: Annual Conference
                  on Neural Information Processing Systems 2017, December 4-9, 2017,
                  Long Beach, CA, {USA}},
  pages        = {5998--6008},
  year         = {2017},
  url          = {https://proceedings.neurips.cc/paper/2017/hash/3f5ee243547dee91fbd053c1c4a845aa-Abstract.html},
  bibsource    = {dblp computer science bibliography, https://dblp.org}
}

@misc{wang2020linformerselfattentionlinearcomplexity,
  author       = {Sinong Wang and
                  Belinda Z. Li and
                  Madian Khabsa and
                  Han Fang and
                  Hao Ma},
  title        = {Linformer: Self-Attention with Linear Complexity},
  journal      = {CoRR},
  volume       = {abs/2006.04768},
  year         = {2020},
  url          = {https://arxiv.org/abs/2006.04768},
  eprinttype    = {arXiv},
  eprint       = {2006.04768},
  bibsource    = {dblp computer science bibliography, https://dblp.org}
}

@inproceedings{choromanski2021rethinkingattentionperformers,
  title        = {Rethinking Attention with Performers},
  author       = {Krzysztof Choromanski and Valerii Likhosherstov and David Dohan and Xingyou Song and Andreea Gane and Tam{\'a}s Sarl{\'o}s and Peter Hawkins and Jared Davis and Afroz Mohiuddin and Lukasz Kaiser and David Belanger and Lucy Colwell and Adrian Weller},
  booktitle    = {International Conference on Learning Representations},
  year         = {2021},
  url          = {https://openreview.net/forum?id=Ua6zuk0WRH}
}

@misc{lee2019settransformerframeworkattentionbased,
  author       = {Juho Lee and
                  Yoonho Lee and
                  Jungtaek Kim and
                  Adam R. Kosiorek and
                  Seungjin Choi and
                  Yee Whye Teh},
  editor       = {Kamalika Chaudhuri and
                  Ruslan Salakhutdinov},
  title        = {Set Transformer: {A} Framework for Attention-based Permutation-Invariant
                  Neural Networks},
  booktitle    = {Proceedings of the 36th International Conference on Machine Learning,
                  {ICML} 2019, 9-15 June 2019, Long Beach, California, {USA}},
  series       = {Proceedings of Machine Learning Research},
  volume       = {97},
  pages        = {3744--3753},
  publisher    = {{PMLR}},
  year         = {2019},
  url          = {http://proceedings.mlr.press/v97/lee19d.html},
  bibsource    = {dblp computer science bibliography, https://dblp.org}
}

@misc{kitaev2020reformerefficienttransformer,
  author       = {Nikita Kitaev and
                  Lukasz Kaiser and
                  Anselm Levskaya},
  title        = {Reformer: The Efficient Transformer},
  booktitle    = {8th International Conference on Learning Representations, {ICLR} 2020,
                  Addis Ababa, Ethiopia, April 26-30, 2020},
  publisher    = {OpenReview.net},
  year         = {2020},
  url          = {https://openreview.net/forum?id=rkgNKkHtvB},
  bibsource    = {dblp computer science bibliography, https://dblp.org}
}

@misc{beltagy2020longformerlongdocumenttransformer,
  author       = {Iz Beltagy and
                  Matthew E. Peters and
                  Arman Cohan},
  title        = {Longformer: The Long-Document Transformer},
  journal      = {CoRR},
  volume       = {abs/2004.05150},
  year         = {2020},
  url          = {https://arxiv.org/abs/2004.05150},
  eprinttype    = {arXiv},
  eprint       = {2004.05150},
  bibsource    = {dblp computer science bibliography, https://dblp.org}
}

@misc{fan2025rectifyingmagnitudeneglectlinear,
  author       = {Qihang Fan and
                  Huaibo Huang and
                  Yuang Ai and
                  Ran He},
  title        = {Rectifying Magnitude Neglect in Linear Attention},
  journal      = {CoRR},
  volume       = {abs/2507.00698},
  year         = {2025},
  url          = {https://doi.org/10.48550/arXiv.2507.00698},
  doi          = {10.48550/ARXIV.2507.00698},
  eprinttype    = {arXiv},
  eprint       = {2507.00698},
  bibsource    = {dblp computer science bibliography, https://dblp.org}
}

@misc{katharopoulos2020transformersrnnsfastautoregressive,
  author       = {Angelos Katharopoulos and
                  Apoorv Vyas and
                  Nikolaos Pappas and
                  Fran{\c{c}}ois Fleuret},
  title        = {Transformers are RNNs: Fast Autoregressive Transformers with Linear
                  Attention},
  booktitle    = {Proceedings of the 37th International Conference on Machine Learning,
                  {ICML} 2020, 13-18 July 2020, Virtual Event},
  series       = {Proceedings of Machine Learning Research},
  volume       = {119},
  pages        = {5156--5165},
  publisher    = {{PMLR}},
  year         = {2020},
  url          = {http://proceedings.mlr.press/v119/katharopoulos20a.html},
  bibsource    = {dblp computer science bibliography, https://dblp.org}
}

@inproceedings{Song_2025, series={SIGIR '25},
  author       = {Xin Song and
                  Xiaochen Li and
                  Jinxin Hu and
                  Hong Wen and
                  Zulong Chen and
                  Yu Zhang and
                  Xiaoyi Zeng and
                  Jing Zhang},
  editor       = {Nicola Ferro and
                  Maria Maistro and
                  Gabriella Pasi and
                  Omar Alonso and
                  Andrew Trotman and
                  Suzan Verberne},
  title        = {{LREA:} Low-Rank Efficient Attention on Modeling Long-Term User Behaviors
                  for {CTR} Prediction},
  booktitle    = {Proceedings of the 48th International {ACM} {SIGIR} Conference on
                  Research and Development in Information Retrieval, {SIGIR} 2025, Padua,
                  Italy, July 13-18, 2025},
  pages        = {2843--2847},
  publisher    = {{ACM}},
  year         = {2025},
  url          = {https://doi.org/10.1145/3726302.3730228},
  doi          = {10.1145/3726302.3730228},
  bibsource    = {dblp computer science bibliography, https://dblp.org}
}

@ARTICLE{5197422,
  author       = {Yehuda Koren and
                  Robert M. Bell and
                  Chris Volinsky},
  title        = {Matrix Factorization Techniques for Recommender Systems},
  journal      = {Computer},
  volume       = {42},
  number       = {8},
  pages        = {30--37},
  year         = {2009},
  url          = {https://doi.org/10.1109/MC.2009.263},
  doi          = {10.1109/MC.2009.263},
  bibsource    = {dblp computer science bibliography, https://dblp.org}
}

@misc{chang2023twintwostagenetworklifelong,
  author       = {Jianxin Chang and
                  Chenbin Zhang and
                  Zhiyi Fu and
                  Xiaoxue Zang and
                  Lin Guan and
                  Jing Lu and
                  Yiqun Hui and
                  Dewei Leng and
                  Yanan Niu and
                  Yang Song and
                  Kun Gai},
  editor       = {Ambuj K. Singh and
                  Yizhou Sun and
                  Leman Akoglu and
                  Dimitrios Gunopulos and
                  Xifeng Yan and
                  Ravi Kumar and
                  Fatma Ozcan and
                  Jieping Ye},
  title        = {{TWIN:} TWo-stage Interest Network for Lifelong User Behavior Modeling
                  in {CTR} Prediction at Kuaishou},
  booktitle    = {Proceedings of the 29th {ACM} {SIGKDD} Conference on Knowledge Discovery
                  and Data Mining, {KDD} 2023, Long Beach, CA, USA, August 6-10, 2023},
  pages        = {3785--3794},
  publisher    = {{ACM}},
  year         = {2023},
  url          = {https://doi.org/10.1145/3580305.3599922},
  doi          = {10.1145/3580305.3599922},
  bibsource    = {dblp computer science bibliography, https://dblp.org}
}

@misc{zhai2024actionsspeaklouderwords,
  author       = {Jiaqi Zhai and
                  Lucy Liao and
                  Xing Liu and
                  Yueming Wang and
                  Rui Li and
                  Xuan Cao and
                  Leon Gao and
                  Zhaojie Gong and
                  Fangda Gu and
                  Jiayuan He and
                  Yinghai Lu and
                  Yu Shi},
  title        = {Actions Speak Louder than Words: Trillion-Parameter Sequential Transducers
                  for Generative Recommendations},
  booktitle    = {Forty-first International Conference on Machine Learning, {ICML} 2024,
                  Vienna, Austria, July 21-27, 2024},
  publisher    = {OpenReview.net},
  year         = {2024},
  url          = {https://openreview.net/forum?id=xye7iNsgXn},
  bibsource    = {dblp computer science bibliography, https://dblp.org}
}

@misc{zhang2024wukongscalinglawlargescale,
  author       = {Buyun Zhang and
                  Liang Luo and
                  Yuxin Chen and
                  Jade Nie and
                  Xi Liu and
                  Shen Li and
                  Yanli Zhao and
                  Yuchen Hao and
                  Yantao Yao and
                  Ellie Dingqiao Wen and
                  Jongsoo Park and
                  Maxim Naumov and
                  Wenlin Chen},
  title        = {Wukong: Towards a Scaling Law for Large-Scale Recommendation},
  booktitle    = {Forty-first International Conference on Machine Learning, {ICML} 2024,
                  Vienna, Austria, July 21-27, 2024},
  publisher    = {OpenReview.net},
  year         = {2024},
  url          = {https://openreview.net/forum?id=8iUgr2nuwo},
  bibsource    = {dblp computer science bibliography, https://dblp.org}
}

@misc{chai2025longerscalinglongsequence,
  author       = {Zheng Chai and
                  Qin Ren and
                  Xijun Xiao and
                  Huizhi Yang and
                  Bo Han and
                  Sijun Zhang and
                  Di Chen and
                  Hui Lu and
                  Wenlin Zhao and
                  Lele Yu and
                  Xionghang Xie and
                  Shiru Ren and
                  Xiang Sun and
                  Yaocheng Tan and
                  Peng Xu and
                  Yuchao Zheng and
                  Di Wu},
  editor       = {M{\'{a}}ria Bielikov{\'{a}} and
                  Pavel Kord{\'{\i}}k and
                  Markus Schedl and
                  Marco de Gemmis and
                  Sole Pera and
                  Rodrigo Alves and
                  Olivier Jeunen and
                  Vito Ostuni},
  title        = {{LONGER:} Scaling Up Long Sequence Modeling in Industrial Recommenders},
  booktitle    = {Proceedings of the Nineteenth {ACM} Conference on Recommender Systems,
                  RecSys 2025, Prague, Czech Republic, September 22-26, 2025},
  pages        = {247--256},
  publisher    = {{ACM}},
  year         = {2025},
  url          = {https://doi.org/10.1145/3705328.3748065},
  doi          = {10.1145/3705328.3748065},
  bibsource    = {dblp computer science bibliography, https://dblp.org}
}

@misc{chen2022efficientlongsequentialuser,
  author       = {Qiwei Chen and
                  Yue Xu and
                  Changhua Pei and
                  Shanshan Lv and
                  Tao Zhuang and
                  Junfeng Ge},
  title        = {Efficient Long Sequential User Data Modeling for Click-Through Rate
                  Prediction},
  journal      = {CoRR},
  volume       = {abs/2209.12212},
  year         = {2022},
  url          = {https://doi.org/10.48550/arXiv.2209.12212},
  doi          = {10.48550/ARXIV.2209.12212},
  eprinttype    = {arXiv},
  eprint       = {2209.12212},
  bibsource    = {dblp computer science bibliography, https://dblp.org}
}

@inproceedings{Si_2024, 
  author       = {Zihua Si and
                  Lin Guan and
                  Zhongxiang Sun and
                  Xiaoxue Zang and
                  Jing Lu and
                  Yiqun Hui and
                  Xingchao Cao and
                  Zeyu Yang and
                  Yichen Zheng and
                  Dewei Leng and
                  Kai Zheng and
                  Chenbin Zhang and
                  Yanan Niu and
                  Yang Song and
                  Kun Gai},
  editor       = {Edoardo Serra and
                  Francesca Spezzano},
  title        = {{TWIN} {V2:} Scaling Ultra-Long User Behavior Sequence Modeling for
                  Enhanced {CTR} Prediction at Kuaishou},
  booktitle    = {Proceedings of the 33rd {ACM} International Conference on Information
                  and Knowledge Management, {CIKM} 2024, Boise, ID, USA, October 21-25,
                  2024},
  pages        = {4890--4897},
  publisher    = {{ACM}},
  year         = {2024},
  url          = {https://doi.org/10.1145/3627673.3680030},
  doi          = {10.1145/3627673.3680030},
  bibsource    = {dblp computer science bibliography, https://dblp.org}
}

@misc{qi2020searchbasedusermodelinglifelong,
  author       = {Qi Pi and
                  Guorui Zhou and
                  Yujing Zhang and
                  Zhe Wang and
                  Lejian Ren and
                  Ying Fan and
                  Xiaoqiang Zhu and
                  Kun Gai},
  editor       = {Mathieu d'Aquin and
                  Stefan Dietze and
                  Claudia Hauff and
                  Edward Curry and
                  Philippe Cudr{\'{e}}{-}Mauroux},
  title        = {Search-based User Interest Modeling with Lifelong Sequential Behavior
                  Data for Click-Through Rate Prediction},
  booktitle    = {{CIKM} '20: The 29th {ACM} International Conference on Information
                  and Knowledge Management, Virtual Event, Ireland, October 19-23, 2020},
  pages        = {2685--2692},
  publisher    = {{ACM}},
  year         = {2020},
  url          = {https://doi.org/10.1145/3340531.3412744},
  doi          = {10.1145/3340531.3412744},
  bibsource    = {dblp computer science bibliography, https://dblp.org}
}

@misc{zhou2018deepevolutionnetworkclickthrough,
  author       = {Guorui Zhou and
                  Na Mou and
                  Ying Fan and
                  Qi Pi and
                  Weijie Bian and
                  Chang Zhou and
                  Xiaoqiang Zhu and
                  Kun Gai},
  title        = {Deep Interest Evolution Network for Click-Through Rate Prediction},
  booktitle    = {The Thirty-Third {AAAI} Conference on Artificial Intelligence, {AAAI}
                  2019, The Thirty-First Innovative Applications of Artificial Intelligence
                  Conference, {IAAI} 2019, The Ninth {AAAI} Symposium on Educational
                  Advances in Artificial Intelligence, {EAAI} 2019, Honolulu, Hawaii,
                  USA, January 27 - February 1, 2019},
  pages        = {5941--5948},
  publisher    = {{AAAI} Press},
  year         = {2019},
  url          = {https://doi.org/10.1609/aaai.v33i01.33015941},
  doi          = {10.1609/AAAI.V33I01.33015941},
  bibsource    = {dblp computer science bibliography, https://dblp.org}
}

@misc{kang2018selfattentivesequentialrecommendation,
  author       = {Wang{-}Cheng Kang and
                  Julian J. McAuley},
  title        = {Self-Attentive Sequential Recommendation},
  booktitle    = {{IEEE} International Conference on Data Mining, {ICDM} 2018, Singapore,
                  November 17-20, 2018},
  pages        = {197--206},
  publisher    = {{IEEE} Computer Society},
  year         = {2018},
  url          = {https://doi.org/10.1109/ICDM.2018.00035},
  doi          = {10.1109/ICDM.2018.00035},
  bibsource    = {dblp computer science bibliography, https://dblp.org}
}

@misc{gui2023hiformerheterogeneousfeatureinteractions,
  author       = {Huan Gui and
                  Ruoxi Wang and
                  Ke Yin and
                  Long Jin and
                  Maciej Kula and
                  Taibai Xu and
                  Lichan Hong and
                  Ed H. Chi},
  title        = {Hiformer: Heterogeneous Feature Interactions Learning with Transformers
                  for Recommender Systems},
  journal      = {CoRR},
  volume       = {abs/2311.05884},
  year         = {2023},
  url          = {https://doi.org/10.48550/arXiv.2311.05884},
  doi          = {10.48550/ARXIV.2311.05884},
  eprinttype    = {arXiv},
  eprint       = {2311.05884},
  bibsource    = {dblp computer science bibliography, https://dblp.org}
}

@inproceedings{wu-etal-2020-mind,
  author       = {Fangzhao Wu and
                  Ying Qiao and
                  Jiun{-}Hung Chen and
                  Chuhan Wu and
                  Tao Qi and
                  Jianxun Lian and
                  Danyang Liu and
                  Xing Xie and
                  Jianfeng Gao and
                  Winnie Wu and
                  Ming Zhou},
  editor       = {Dan Jurafsky and
                  Joyce Chai and
                  Natalie Schluter and
                  Joel R. Tetreault},
  title        = {{MIND:} {A} Large-scale Dataset for News Recommendation},
  booktitle    = {Proceedings of the 58th Annual Meeting of the Association for Computational
                  Linguistics, {ACL} 2020, Online, July 5-10, 2020},
  pages        = {3597--3606},
  publisher    = {Association for Computational Linguistics},
  year         = {2020},
  url          = {https://doi.org/10.18653/v1/2020.acl-main.331},
  doi          = {10.18653/V1/2020.ACL-MAIN.331},
  bibsource    = {dblp computer science bibliography, https://dblp.org}
}

@inproceedings{liu2025recflow,
  author       = {Qi Liu and
                  Kai Zheng and
                  Rui Huang and
                  Wuchao Li and
                  Kuo Cai and
                  Yuan Chai and
                  Yanan Niu and
                  Yiqun Hui and
                  Bing Han and
                  Na Mou and
                  Hongning Wang and
                  Wentian Bao and
                  Yunen Yu and
                  Guorui Zhou and
                  Han Li and
                  Yang Song and
                  Defu Lian and
                  Kun Gai},
  title        = {RecFlow: An Industrial Full Flow Recommendation Dataset},
  booktitle    = {The Thirteenth International Conference on Learning Representations,
                  {ICLR} 2025, Singapore, April 24-28, 2025},
  publisher    = {OpenReview.net},
  year         = {2025},
  url          = {https://openreview.net/forum?id=vVHc8bGRns},
  bibsource    = {dblp computer science bibliography, https://dblp.org}
}

@misc{guo2024embeddingcollapsescalingrecommendation,
  author       = {Xingzhuo Guo and
                  Junwei Pan and
                  Ximei Wang and
                  Baixu Chen and
                  Jie Jiang and
                  Mingsheng Long},
  title        = {On the Embedding Collapse when Scaling up Recommendation Models},
  booktitle    = {Forty-first International Conference on Machine Learning, {ICML} 2024,
                  Vienna, Austria, July 21-27, 2024},
  publisher    = {OpenReview.net},
  year         = {2024},
  url          = {https://openreview.net/forum?id=aPVwOAr1aW},
  bibsource    = {dblp computer science bibliography, https://dblp.org}
}
